\documentclass[twoside,journal]{IEEEtran}
\usepackage{amsfonts}
\usepackage{amssymb}
\usepackage{amsmath}
\usepackage{mathrsfs}
\usepackage{verbatim}
\usepackage{subfigure}
\usepackage{balance}
\usepackage{booktabs}
\usepackage{fancybox}
\usepackage{bm}
\usepackage{extarrows}
\usepackage{algorithm}
\usepackage{algorithmic}
\usepackage{multirow}
\usepackage{array}
\usepackage{graphicx}
\usepackage{epstopdf}
\newtheorem{theorem}{Theorem}
\newtheorem{lemma}{Lemma}
\newtheorem{corollary}{Corollary}
\newtheorem{proof}{Proof}
\newtheorem{proposition}{Proposition}
\newtheorem{property}{Property}

\begin{document}

\title{Physical Layer Security in Heterogeneous\\
Cellular Networks}
\author{Hui-Ming~Wang,~\IEEEmembership{Member,~IEEE,~}
        Tong-Xing~Zheng,~\IEEEmembership{Student Member,~IEEE}\\
        Jinhong~Yuan,~\IEEEmembership{Fellow,~IEEE,~}
        Don~Towsley,~\IEEEmembership{Fellow,~IEEE,~}
        and~Moon~Ho~Lee
\thanks{H.-M. Wang and T.-X. Zheng are with the School of Electronic and Information Engineering, and also with the MOE Key Lab for Intelligent Networks and Network Security,
Xi'an Jiaotong University, Xi'an, 710049, Shaanxi, China. Email: {\tt xjbswhm@gmail.com},
{\tt txzheng@stu.xjtu.edu.cn}.
}
\thanks{J. Yuan is with the School of Electrical Engineering and Telecommunications, University of New South Wales, Sydney, Australia. Email: {\tt j.yuan@unsw.edu.au}.}
\thanks{D. Towsley is with the Department of Computer Science, University of Massachusetts, Amherst, MA, US. Email: {\tt towsley@cs.umass.edu}.}
\thanks{M. H. Lee is with the Division of Electronics Engineering, Chonbuk National 	University, Jeonju 561-756, Korea. Email: {\tt moonho@jbnu.ac.kr}.}
}

\maketitle
\vspace{-0.8 cm}

\begin{abstract}
    The heterogeneous cellular network (HCN) is a promising approach to the deployment of 5G cellular networks.
    This paper comprehensively studies physical layer security in a multi-tier HCN where base stations (BSs), authorized users and eavesdroppers are all randomly located.
    We first propose an access threshold based secrecy mobile association policy that associates each user with the BS providing the maximum \emph{truncated average received signal power} beyond a threshold.
    Under the proposed policy, we investigate the connection probability and secrecy probability of a randomly located user, and provide tractable expressions for the two metrics.
    Asymptotic analysis reveals that setting a larger access threshold increases the connection probability while decreases the secrecy probability.
    We further evaluate the network-wide secrecy throughput and the minimum secrecy throughput per user with both connection and secrecy probability constraints.
    We show that introducing a properly chosen access threshold significantly enhances the secrecy throughput performance of a HCN.
\end{abstract}


\IEEEpeerreviewmaketitle

\section{Introduction}

\IEEEPARstart{T}{he} deployment of heterogeneous cellular networks (HCNs) is a promising approach to providing seamless wireless coverage and high network throughput in 5G mobile communication.
A HCN deploys a variety of infrastructure, such as macro, pico, and  femto base stations (BSs), as well as fixed relay  stations in different tiers \cite{Lagrange1997Multitier}.
BSs in different tiers have different transmit powers and coverages.
For example, a macrocell uses the highest power to provide large coverage, while a femtocell is usually a low-power home BS intended for short-range communications.
Due to the co-channel spectrum sharing between different tiers,  network interference in the HCN is much more severe than that in a conventional single-tier cellular network, thus posing a challenge to the successful co-existence of tiers \cite{Xia2010Open}.
Therefore, one of the major challenges in deploying HCNs is to efficiently manage network interference.
Femtocell access control, using either closed or open access, is an important mechanism for interference management \cite{Roche2010Access}.
In closed access, femtocell access points provide service only to the specified subscribers, whereas arbitrary nearby users can use the femtocell in open access.
Xia \emph{et al.} \cite{Xia2010Open} point out that open access is preferred by network operators, since it not only efficiently reduces cross-tier interference, but also provides an inexpensive way to expand network capacity.

However, due to the open system architecture of a HCN and the broadcast nature of wireless communications, information transmissions intended for authorized user equipments (UEs) are more vulnerable to eavesdroppers (also named unauthorized users).
As shown in Fig. \ref{HCN}, eavesdroppers (Eves) find it easy to overhear legitimate communications.
Therefore, secure transmission is a significant concern when designing HCNs.
Unfortunately, the existing literature on HCNs has mainly focused on network throughput and energy efficiency; little of it has involved security issues.
\begin{figure}[!t]
\centering
\includegraphics[width=3.0in]{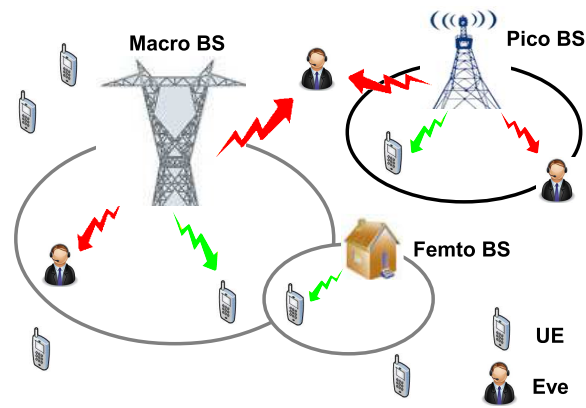}
\caption{A 3-tier macro/pico/femto HCN where authorized users coexist with eavesdroppers.}
\label{HCN}
\end{figure}

Physical layer security (PLS), or, information-theoretic security, has drawn ever-increasing attention since Wyner's seminal research \cite{Wyner1975Wire-tap}, where he introduced the degraded wiretap channel model and defined the concept of \emph{secrecy capacity}.
During the past decades, the wiretap channel model has been generalized to multi-input multi-output (MIMO) channels \cite{Shafiee2009Towards}-\cite{Khisti2010Secure2}, cooperative relay channels \cite{Dong2010Improving}-\cite{Zheng2015Outage}, and two-way channels \cite{Tekin2008General}-\cite{Wang2012Hybrid}, {etc}.
A large number of secrecy transmission techniques
and schemes have been proposed for wireless communications (refer to \cite{Mukherjee2014Principles} and its references).

\subsection{Related Works and Motivation}
Early research on PLS has focused on point-to-point links or single-cell scenarios.
In some works, Eve's channel state information (CSI) is assumed to be perfectly known at the transmitter, which is clearly not practical in real wiretap scenarios, since Eves are usually passive.
Without Eve's CSI, Goel \emph{et al.} \cite{Goel2008Guaranteeing} proposed a multi-antenna transmission strategy with artificial noise embedded into information signals to confuse Eve.
This method has become a popular approach to enhance the PLS, and has attracted a stream of research, e.g., \cite{Zhou2010Secure}-\cite{WangANAFF}.
The idea of artificial noise has also been extended to relay systems with jammers, in which cooperative jamming techniques \cite{Wang2012Hybrid}, \cite{Wang2015Hybrid}, \cite{Deng2015Secrecy} have been proposed to improve PLS.
However, due to dynamic and large-scale wireless network topologies, the spatial positions of network nodes and propagation path losses become very critical factors influencing secrecy performance, which unfortunately has been considered by none of the above endeavors.

Recently, stochastic geometry theory has provided a powerful tool to study the average behavior of a network by modeling the positions of network nodes according to a spatial distribution such as a Poisson point process (PPP) \cite{Haenggi2009Stochastic}.
Under a stochastic geometry framework, authors in \cite{Zhou2011Secure}-\cite{Zheng2015Multi} studied the secure multi-antenna transmission against PPP distributed Eves.
More specifically, Zhou \emph{et al.} \cite{Zhou2011Secure} evaluated the secure connectivity of two multi-antenna transmission techniques: a directional antenna scheme and an eigen-beamforming scheme.
Zheng \emph{et al.} \cite{Zheng2014Transmission} investigated the average secrecy outage probability in a multi-input single-output (MISO) wiretap channel for both non-colluding and colluding Eves.
Ghogho \emph{et al.} \cite{Ghogho2011Physical} derived the probability of a positive secrecy rate achieved by the artificial-noise method in MIMO channels.
Zheng \emph{et al.} \cite{Zheng2015Multi}, \cite{Zheng2015Multi2} proposed both dynamic and static parameter design schemes for the artificial-noise-aided transmission to maximize secrecy throughput subject to a secrecy outage probability constraint.

Research on PLS has been further extended to ad hoc networks \cite{Zhou2011Throughput}, \cite{Zhang2013Enhancing} and cellular networks \cite{Wang2013Physical}, \cite{Geraci2014Physical}, where the placement of transmitters and receivers are both modeled as PPPs.
Zhou, \emph{et al.} \cite{Zhou2011Throughput}, \cite{Zhang2013Enhancing} considered single- and multi-antenna transmissions in an ad hoc network, and  provided a tradeoff analysis between connectivity and secrecy, and further measured the secrecy transmission capacity.
Wang \emph{et al.} \cite{Wang2013Physical} evaluated the secrecy performance of cellular networks considering the cell association and information exchange between BSs, and provided tractable results for the achievable secrecy rate under different assumptions on the information of Eves' locations.
However, they only considered a single-antenna case, ignoring both small-scale fading and inter-cell interference.
This work has been extended by \cite{Geraci2014Physical}, where Geraci \emph{et al.} investigated the average secrecy rate utilizing the regularized channel inversion transmit precoding from a perspective of massive MIMO systems.

Due to the multi-tier hierarchical architecture, HCNs bring new challenges to the investigation of PLS compared with the conventional single-tier topology \cite{Yang2015Safeguarding}.
In addition to cross-cell interference, HCNs introduce severe cross-tier interference.
Both reliability and secrecy of data transmissions should be taken into account, which makes analyzing the impact of interference on both UEs and Eves much more complicated, especially when system parameters differ between different tiers.
Besides, mobile terminals can access an arbitrary tier, e.g., open access, which calls for specific mobile association policies that consider both quality of service (QoS) and secrecy.

A very recent contribution \cite{Lv2015Secrecy} considered PLS in a two-tier heterogeneous network with one Eve wiretapping macrocell users.
We point out that, the authors in \cite{Lv2015Secrecy} focused on the design/optimization of secrecy beamforming, but not from the perspective of network analysis and deployment.
Their conclusions are based on the idealized assumption that the CSI of the Eve is perfectly available.
Moreover, they considered neither the multi-Eve wiretap scenarios, nor the random spatial positions of network nodes and the large-scale path loss.
To the best of our knowledge, no prior work has accounted for PLS when designing HCNs, and a fundamental analysis framework to evaluate the secrecy performance in HCNs is lacking, which has motivated our work.

\subsection{Our Work and Contributions}
In this paper, we extend PLS to a $K$-tier HCN where the positions of BSs, UEs and Eves are all modeled as independent homogeneous PPPs.
We provide a comprehensive performance analysis of artificial-noise-aided multi-antenna secure transmission under a stochastic geometry framework.
Our main contributions are summarized as follows:

\textbf{i})
We propose a secrecy mobile association policy based on the \emph{truncated average received signal power} (ARSP).
Specifically, a typical UE is only permitted to associate with the BS providing the highest ARSP; if the highest ARSP is below a pre-set access threshold, the UE remains inactive.
 We derive closed-form expressions for the tier association probability for this policy (the probability that a tier is associated with the typical UE) and the BS activation probability (the probability that a BS associates at least one UE), which are essential to analyze the key performance metrics.

\textbf{ii})
We analyze the connection probability of a randomly located UE, which is defined as the probability that the signal-to-interference-plus-noise ratio (SINR) of the UE lies above a target SINR.
We derive a new accurate integral representation of the connection probability and an analytically tractable expression under the interference-limited case.
An asymptotic analysis of the connection probability reveals that setting a larger access threshold is beneficial for improving link quality.

\textbf{iii})
We analyze the user secrecy probability, which is defined as the probability that the SINR of an arbitrary Eve lies below a SINR threshold.
We derive analytical upper and lower bounds for the secrecy probability, which are close to the exact values in the high secrecy probability region.
We find that the access threshold, BS density and power allocation ratio respectively displays a tradeoff between the connection and secrecy probabilities, and that these parameters should be carefully designed to balance link quality and secrecy.

\textbf{iv})
We investigate network-wide secrecy throughput subject to connection and secrecy probability constraints.
We derive closed-form expressions for the rate of redundant information in small-antenna and large-antenna cases, respectively.
We further evaluate the minimum per user secrecy throughput.
We show that, compared with non-threshold mobile access, our threshold-based policy can significantly increase secrecy throughout when the access threshold is properly chosen.

Leveraging the obtained analytical expressions, we provide various tractable predictions of network performance and guidelines for future network designs.
For instance, setting a larger access threshold helps to improve link quality.
However, if we aim to increase network-wide secrecy throughput, we should properly choose the access threshold, but not set it as large as possible.

\subsection{Organization and Notations}
The remainder of this paper is organized as follows.
In Section II, we describe the system model.
In Sections III and IV, we investigate the connection probability and secrecy probability, respectively.
In Section V, we evaluate the network-wide secrecy throughput.
In Section VI, we conclude our work.

\emph{Notations}:
bold uppercase (lowercase) letters denote matrices (column vectors).
$(\cdot)^{\dag}$, $(\cdot)^{\mathrm{T}}$, $|\cdot|$, $\|\cdot\|$, $\mathbb{P}\{\cdot\}$, and $\mathbb{E}_A(\cdot)$ denote conjugate, transpose, absolute value, Euclidean norm, probability, and expectation with respect to (w.r.t.) $A$, respectively.
$\mathcal{CN}(\mu, \sigma^2)$, $\mathrm{Exp}(\lambda)$ and $\Gamma(N,\lambda)$ denote the circularly symmetric complex Gaussian distribution with mean $\mu$ and variance $\sigma^2$, exponential distribution with parameter $\lambda$, and gamma distribution with parameters $N$ and $\lambda$, respectively.
$\mathbb{R}^{m\times n}$ and $\mathbb{C}^{m\times n}$ denote the $m\times n$ real and complex number domains, respectively.
$\log(\cdot)$ and $\ln(\cdot)$ denote the base-2 and natural logarithms, respectively.
$f_V(\cdot)$ and $F_V(\cdot)$ denote the probability density function (PDF) and cumulative distribution function (CDF) of a random variable $V$, respectively.
$\mathcal{B}(o,r)$ describes a disk with center $o$ and radius $r$.
$\textsf{B}_{z,k}$ and $\textsf{U}_{x}$ ($\textsf{E}_{x}$) represent a BS at location $z$ in tier $k$ and a UE (Eve) at location $x$, respectively.
$[x]^{+}\triangleq \max(x,0)$ with $x$ a real number.
 $C_{\alpha,m}\triangleq \frac{\Gamma\left(m-1+{\frac{2}{\alpha}}\right)
\Gamma\left(1-{\frac{2}{\alpha}}\right)}
{\Gamma\left(m-1\right)}$ for $m\ge 2$.


\section{System Model}
We consider a $K$-tier HCN where the BSs in different tiers have different operating parameters (e.g., transmit power and antenna numbers), while those in the same tier share the same parameters.
Define $\mathcal{K}\triangleq \{1,2,\cdots,K\}$.
In tier $k$, the BSs are spatially distributed according to a homogeneous PPP $\Phi_k$ with density $\lambda_k$ in a two-dimensional plane $\mathbb{R}^2$.
As depicted in Fig. \ref{HCN}, there coexist UEs and Eves, where the UEs are legitimate destinations while the Eves are wiretappers attempting to intercept the secret information intended for the UEs.
The locations of the UEs and Eves are characterized by two independent homogeneous PPPs $\Phi_u$ and $\Phi_e$ with densities $\lambda_u$ and $\lambda_e$, respectively.

\subsection{Channel Model}
Wireless channels in the HCN are assumed to undergo flat Rayleigh fading together with a large-scale path loss governed by the exponent $\alpha>2$ \footnote{The analysis of different $\alpha$'s in different tiers can be performed in a similar way, which is omitted in this paper for tractability. }.
Each BS in tier $k$ has $M_k$ antennas, and UEs and Eves are each equipped with a single antenna.
 The channel from $\textsf{B}_{z,k}$ to $\textsf{U}_{x}$ or $\textsf{E}_{x}$ is characterized by $\mathbf{h}_{zx}r_{zx}^{-\frac{\alpha}{2}}$, where $\mathbf{h}_{zx}\in\mathbb{C}^{M_k\times 1}$ denotes the small-scale fading vector with independent and identically distributed (i.i.d.) entries $h_{zx,j}\thicksim \mathcal{CN}(0,1)$, and $r_{zx}$ denotes the path distance.
The noise at each receive node is $n_x\sim\mathcal{CN} (0, N_0)$.
We assume that each BS knows the CSIs of its associated UEs.
Since each Eve passively receives signals, its CSI is unknown, whereas its channel statistics information is available\footnote{This assumption is very generic and has been extensively adopted in the literature on PLS, e.g., \cite{Goel2008Guaranteeing}-\cite{Zhang2013Design}, \cite{Zhou2011Secure}-\cite{Geraci2014Physical}.}.

\subsection{Wyner's Wiretap Code}
We utilize the well-known Wyner's wiretap encoding scheme \cite{Wyner1975Wire-tap} to encode secret information.
Let $\mathcal{R}_{t,k}$ and $\mathcal{R}_{e,k}$ denote respectively the rates of the transmitted codewords and redundant information (to protect from eavesdropping) for tier $k$, and $\mathcal{R}_{s,k}=\mathcal{R}_{t,k}-\mathcal{R}_{e,k}$ denotes the secrecy rate.
Consider a typical legitimate BS-UE pair in tier $k$. 
If the channel from the BS to the UE can support the rate $\mathcal{R}_{t,k}$,
the UE is able to decode the secret messages, which corresponds to a \emph{reliable connection} event.
If none of the channels from the BS to the Eves can support the redundant rate $\mathcal{R}_{e,k}$, the information is deemed to be protected against wiretapping, i.e., \emph{secrecy} is achieved \cite{Zhang2013Enhancing}.

\subsection{Artificial-Noise-Aided Transmission}
To deliberately confuse Eves while guaranteeing reliable links to UEs, each BS employs the artificial-noise-aided transmission strategy \cite{Goel2008Guaranteeing}.
The transmitted signal of $\textsf{B}_{z,k}$ is designed in the form of
\begin{equation}\label{x_1j}
  \mathbf{x}_{z}=\sqrt{\phi_k P_k}\mathbf{w}_{z}s_{z}+\sqrt{{(1-\phi_k)P_k}}
  \mathbf{W}_{z}\mathbf{v}_{z},~z\in\Phi_k,
\end{equation}
where $s_z$ is the information-bearing signal with $\mathbb{E}[|s_z|^2]=1$, $\mathbf{v}_{z}\in\mathbb{C}^{(M_k-1)\times1}$ is an artificial noise vector with i.i.d. entries ${v}_{z,i}\sim\mathcal{CN}\left({0},
\frac{1}{M_k-1}\right)$, and $\phi_k\in[0,1]$ denotes the power allocation ratio of the information signal power to the total transmit power $P_k$.
$\mathbf{w}_{z}={{\mathbf{h}}_{z}^{\dag}}
/{\|{\mathbf{h}}_{z}\|}$ is the beamforming vector for the served UE, with $\mathbf{h}_{z}$ the corresponding channel.
$\mathbf{W}_{z}\in\mathbb{C}^{M_k\times (M_k-1)}$ is a weight matrix for the artificial noise, and the columns of $\mathbf{W}\triangleq [\mathbf{w}_{z} ~\mathbf{W}_{z}]$ constitute an orthogonal basis.

\subsection{Secrecy Mobile Association Policy}
\begin{figure}[!t]
\centering
\includegraphics[width=3.0in]{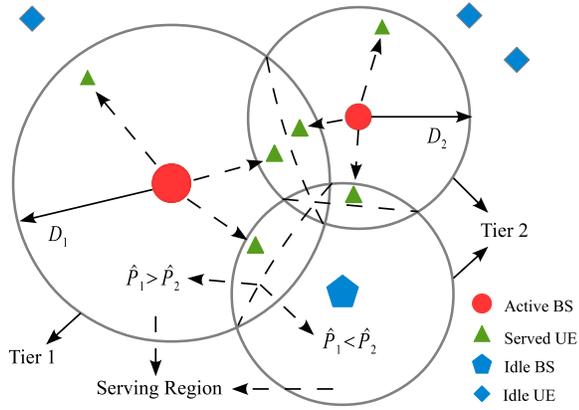}
\caption{An illustration of our mobile association policy in a 2-tier HCN.
A UE connects to the BS providing the highest ARSP instead of the nearest BS.
Those UEs outside the serving regions of BSs can not be served. A BS remains idle if it has no UE to serve.}
\label{Serving Region}
\end{figure}

We consider an open-access system where each UE is allowed to be served by a BS from an arbitrary tier, and it is associated with the tier that provides the largest ARSP.
For an arbitrary UE, the ARSP related to tier $k$ is defined as $\hat{P}_k \triangleq {P_k M_k}{r_k^{-\alpha}}$, where $r_k$ denotes the distance from the UE to the nearest BS in tier $k$.
To avoid access with too low an ARSP, we propose a \emph{truncated ARSP} based mobile association policy, where we introduce an \emph{access threshold} $\tau$, and only allow those UEs with ARSPs larger than $\tau$ to be associated.
Mathematically, the truncated ARSP is defined as
\begin{align}\label{arsp}
\hat{P}_k=\begin{cases}
~{P_k M_k}{r_k^{-\alpha}}, & r_k\le D_k,\\
~0, & r_k>D_k,
\end{cases}
\end{align}
where $D_k=\left(\frac{P_kM_k}{\tau}\right)
^{\frac{1}{\alpha}}$ denotes the radius of the \emph{serving region} of an arbitrary BS in tier $k$, and the index of the tier to which the considered UE is associated is determined by
\begin{align}\label{bs_ass}
  n^* = \arg\max_{k\in\mathcal{K}
  }\hat{P}_k.
\end{align}

Our mobile association policy is illustrated in Fig. \ref{Serving Region}, where the serving region of a BS has been clearly shown.
Due to tiers having different operating parameters, the average coverage regions of each cell do not correspond to a standard Voronoi tessellation, but closely resemble a \emph{circular Dirichlet tessellation} \cite{Dhillon2012Modeling}.
It is worth mentioning that compared to conventional non-threshold association policies,
our threshold-based policy benefits secrecy in the following two aspects:

1) It restrains a BS from associating the UEs outside its serving region (the ARSP outside the serving region is always inferior to that inside), such that not only a good link quality can be guaranteed but also more power can be used to transmit artificial noise to degrade the wiretap channels.

2) If a BS serves no UE, it is kept idle in order to reduce both intra- and cross-tier interference.
 Therefore, the link quality for the active BSs consequently improves, which has the potential of increasing secrecy rates or secrecy throughput.

The proposed mobile association policy is quite applicable to the HCN with secrecy requirements.
We will see in subsequent analysis that $\tau$ plays a critical role in secrecy transmissions.
Before going into the analysis, we first define and compute \emph{tier association probability} \cite{Jo2012Heterogeneous} and \emph{BS activation probability}, which are essential for analyzing our key performance metrics in the sequel.

For ease of notation, we define $\delta\triangleq {2}/{\alpha}$, $\Xi\triangleq \sum_{j\in\mathcal{K}}\lambda_j(P_jM_j)^{{\delta}}$, and
$\mathcal{C}_{j,k} \triangleq \frac{\mathcal{C}_j}{\mathcal{C}_k}$, $\forall \mathcal{C}\in\{P,M,\lambda,\phi\}$.

Recalling \eqref{bs_ass}, the association probability of tier $k$ is mathematically defined as
\begin{equation}\label{sk_def}
  \mathcal{S}_k\triangleq\mathbb{P}\{n^*=k\}=\mathbb{P}\{ \hat{P}_k>\hat{P}_j, \forall j\in\mathcal{K}\setminus k\}.
\end{equation}
It has a closed-form expression provided by the following lemma.
\begin{lemma}\label{association_probability_lemma}
\textit{
The association probability of tier $k$ is given by
\begin{equation}\label{S_k}
  \mathcal{S}_k = \lambda_k(P_kM_k)^{{\delta}} \Xi^{-1}
    \left(1-e^{-\pi\tau^{-{\delta}}\Xi}\right).
\end{equation}}
\end{lemma}
\begin{proof}
    Please see Appendix \ref{association_probability_proof}.
\end{proof}

From Lemma \ref{association_probability_lemma}, we make the following three observations:

1) Tiers with large BS densities, high transmit power, and more BS antennas are more likely to have UEs associated with them.
When tier $k$ has a much larger $\lambda_k$, $P_k$, or $M_k$ than other tiers, $\mathcal{S}_k$ can be approximated by $1-e^{-\pi D_k^2\lambda_k}$, and converges to one as $\lambda_k$ ($P_k$ or $M_k$) goes to infinity.

2) Due to the restriction of $\tau$, an arbitrary UE has a probability
$\mathcal{S}=\sum_{k\in\mathcal{K}}\mathcal{S}_k
=1-e^{-\pi\tau^{-{\delta}}\Xi}$ of being associated with a BS, which implies it has probability $e^{-\pi\tau^{-{\delta}}\Xi}$ of being idle.
This differs from a non-threshold mobile association policy \cite{Jo2012Heterogeneous} which always associates a UE with a tier.

3) Each BS can associate with multiple UEs, and the average number of UEs per BS in tier $k$, i.e., cell load, is $\mathcal{N}_k$ = $\frac{\lambda_u}{\lambda_k}\mathcal{S}_k$ (see \cite[Lemma 2]{Jo2012Heterogeneous}).

We assume that a BS utilizes time division multiple access (TDMA) to efficiently eliminate intra-cell interference.
Due to the overlap of serving regions among cells and different biases towards admitting UEs, even if a BS has UEs located within its serving region, it is inactive when all these UEs are associated with the other BSs (see the idle BS in Fig. \ref{Serving Region}).
We define the \emph{BS activation probability} of tier $k$ as
\begin{equation}\label{A_k_def}
  \mathcal{A}_k \triangleq \mathbb{P}\{\textit{A BS in tier } k \textit{ associates with at least one UE} \}.
\end{equation}
It has a closed-form expression given in the following lemma.
\begin{lemma}\label{activation_probability_lemma}
\textit{The BS activation probability of tier $k$ is given by
\begin{equation}\label{A_k}
  \mathcal{A}_k = 1-\exp\left(-\lambda_u(P_kM_k)^{{\delta}}
  \Xi^{-1}\left(1-e^{-\pi\tau^{-{\delta}}\Xi}\right)\right).
\end{equation}}
\end{lemma}
\begin{proof}
Please see Appendix \ref{activation_probability_proof}.
\end{proof}

    The BS activation probability is important for analyzing HCNs since the level of intra- and cross-tier interferences depends heavily on it.
    Although some empirical approximations of the BS activation probability have been given for conventional cellular networks, e.g., \cite{Li2014Throughput}, they do not apply to HCNs.
    Therefore, our derivation of the BS activation probability is essential for our analysis.

From Lemmas \ref{association_probability_lemma} and \ref{activation_probability_lemma}, we obtain $\mathcal{A}_k = 1-e^{-\frac{\lambda_u}{\lambda_k}\mathcal{S}_k}$.
Obviously, BSs with higher power and more antennas have higher activation probabilities.
Since $\frac{\lambda_u}{\lambda_k}\mathcal{S}_k$ monotonically decreases in $\lambda_k$ (see \eqref{S_k}), it is easy to prove that $\mathcal{A}_k$ is a monotonically decreasing function of  $\lambda_k$, which indicates that deploying more cells results in a smaller BS activation possibility.
On the contrary, introducing more UEs (a larger $\lambda_u$) increases $\mathcal{A}_k$.
We can also readily prove that both $\mathcal{S}_k$ and $\mathcal{A}_k$ decrease in $\tau$, just as validated in Fig. \ref{SA}.
\begin{figure}[!t]
\centering
\includegraphics[width=3.0in]{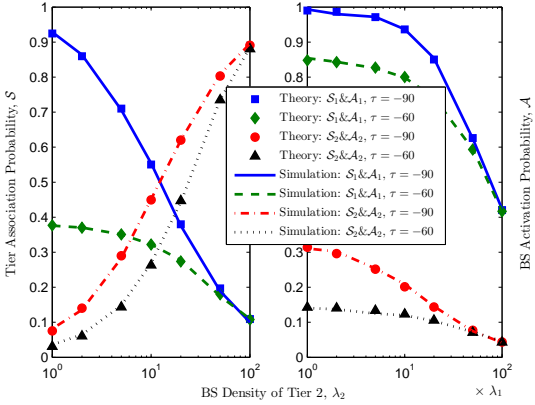}
\caption{Tier association probability and BS activation probability in a 2-tier HCN vs. $\lambda_2$ for different $\tau$(dBm)'s.}
\label{SA}
\end{figure}
The set of active BSs in tier $k$ is a thinning of $\Phi_k$, denoted by $\Phi_k^o$, with density $\lambda_k^{o}=\mathcal{A}_k\lambda_k$.
We have the following property w.r.t. the new $\lambda_k^{o}$.
\begin{property}\label{BS_density_property}
\textit{$\lambda_k^{o}$  monotonically increases in $\lambda_k$, and tends to $\lambda_u$ as $\lambda_k$ goes to infinity.}
\end{property}
\begin{proof}
    Pleas see Appendix \ref{BS_density_proof}.
\end{proof}

This property indicates that although deploying more cells results in smaller BS activation probabilities, it actually increases the number of active BSs.
We emphasize that many existing works, e.g., \cite{Dhillon2012Modeling}, \cite{Jo2012Heterogeneous}, implicitly assume that all BSs are activated, which is not appropriate  mathematically, since the total number of active BSs is limited by the number of UEs under a one-to-one service mode.
This inevitably leads to an inaccurate evaluation of network performance.
In this paper, the BS activation probability is taken into consideration, and $\lambda_k^o<\lambda_u$ strictly holds, which is more realistic compared with \cite{Dhillon2012Modeling}, \cite{Jo2012Heterogeneous}.

In the following sections, we provide a comprehensive analysis of PLS in HCNs incorporating the connection probability, secrecy probability, and network-wide secrecy throughput, respectively.
We stress that, due to secrecy considerations, the analysis is fundamentally different from the existing works without secrecy constraints.
By deriving various analytical expressions for our performance metrics, we aim to provide tractable predictions of network performance and guidelines for future network designs.

\section{User Connection Probability}
In this section, we investigate the connection probability of a randomly located UE.
Here connection probability corresponds to the probability that a secret message is decoded by this UE.

Without lose of generality, we consider a typical UE located at the origin $o$ and served by $\textsf{B}_{b,k}$.
In addition to the desired information signal from the serving BS $\textsf{B}_{b,k}$, $\textsf{U}_{o}$ receives interference consisting of both undesired information signals and artificial noise from the BSs 1) in tier $k$ (except for $\textsf{B}_{b,k}$), and 2) in all the other tiers.
The received signal at $\textsf{U}_{o}$ is given by
\begin{align}\label{yo}
    &y_{o}
    = \underbrace{\frac{\sqrt{\phi_k P_k}
    \mathbf{h}_{b}^\mathrm{T}\mathbf{w}_{b}s_{b}}
    {R_k^{\alpha/2}}}_{\mathrm{information~signal}}+\nonumber\\
    &\underbrace{\sum_{j\in\mathcal{K}}
    \sum_{z\in\Phi_j^o
    \setminus b} \frac{\sqrt{\phi_j P_j}\mathbf{h}_{zo}^{\mathrm{T}}
    \mathbf{w}_{z}s_{z}+\sqrt{(1-\phi_j) P_j}\mathbf{h}_{zo}^{\mathrm{T}}
    \mathbf{W}_{z}\mathbf{v}_{z}}{r_{zo}^{\alpha/2}}}
    _{\mathrm{intra- ~and ~ cross-tier~interference~ (signal~ and~ artificial~ noise)}}+n_{o},
\end{align}
where $R_k$ represents the distance between $\textsf{U}_{o}$ and $\textsf{B}_{b,k}$.

\subsection{General Result}
The connection probability of $\textsf{U}_{o}$ associated with tier $k$ is defined as the probability of the event that the instantaneous SINR of $\textsf{U}_{o}$ exceeds or equals a target SINR $\beta_{t}$, i.e.,
\begin{equation}\label{pc_k}
    \mathcal{P}_{c,k}\triangleq\mathbb{P}\{\mathbf{SINR}_{o,k}
    \ge\beta_{t}\},
\end{equation}
where $\mathbf{SINR}_{o,k}$ is given by
\begin{equation}\label{sinr_o}
    \mathbf{SINR}_{o,k} = \frac{\phi_k P_k\|\mathbf{h}_{b}\|^2 R_k^{-\alpha}}
    { \sum_{j\in\mathcal{K}} I_{jo}+N_0},
\end{equation}
with
$I_{jo}=\sum_{z\in\Phi^o_j\setminus b}\frac{\phi_jP_j \left(|\mathbf{h}_{zo}^{\mathrm{T}}\mathbf{w}_{z}|^2
+\xi_j\|\mathbf{h}_{zo}^{\mathrm{T}}
\mathbf{W}_{z}\|^2\right)}
{r_{zo}^{\alpha}}$ and $\xi_j\triangleq\frac{\phi_j^{-1}-1}{M_j-1}$.
We note from \eqref{bs_ass} that there should be an exclusion region $\mathcal{B}\left(o,\left({P_{j,k}M_{j,k}}\right)
^{\frac{1}{\alpha}}R_k\right)$ around $\textsf{U}_o$ for tier $j\in\mathcal{K}$; all interfering BSs in tier $j$ are located outside of this region.

Let $I_o=\sum_{j\in\mathcal{K}} I_{jo}$ and $s \triangleq \frac{ R_k^{\alpha}\beta_{t}}{\phi_kP_k}$.
$\mathcal{P}_{c,k}$ can be calculated by substituting \eqref{sinr_o} into \eqref{pc_k}
\begin{align}\label{pc_k1}
\mathcal{P}_{c,k}& = \mathbb{E}_{R_k}\mathbb{E}_{I_o} \left[ \mathbb{P}\left\{\|\mathbf{h}_{b}\|^2\geq {s}(I_o+N_0)\right\}\right]\nonumber\\
   &\!\! \!\stackrel{\mathrm{(a)}}
    = \mathbb{E}_{R_k} \mathbb{E}_{I_o}\left[e^{-s(I_o+N_0)}
    \sum_{m=0}^{M_k-1}\frac{s^m(I_o+N_0)^m}
    {m!}\right]\nonumber\\
    &\!\!\!= \sum_{m=0}^{M_k-1}\mathbb{E}_{R_k}\mathbb{E}_{I_o} \left[e^{-sN_0}\sum_{p=0}^{m}\binom{m}{p}\frac{N_0^{m-p}s^me^{-sI_o}}
    {m!}I_o^p\right]\nonumber\\
    &\!\!\!\stackrel{\mathrm{(b)}}= \sum_{m=0}^{M_k-1}\mathbb{E}_{R_k}\left[e^{-sN_0}
     \sum_{p=0}^{m}\binom{m}{p}\frac{(-1)^pN_0^{m-p}s^m}
    {m!}\mathcal{L}^{(p)}_{I_o}(s)\right],
\end{align}
where (a) holds for $\|\mathbf{h}_{b}\|^2\sim \Gamma(M_k,1)$, and (b) is obtained from \cite[Theorem 1]{Hunter08Transmission} with $\mathcal{L}^{(p)}_{I_o}(s)$ the $p$-order derivative of the Laplace transform $\mathcal{L}_{I_o}(s)$ evaluated at $s$.

The major difficulty in calculating $\mathcal{P}_{c,k}$ is calculating $\mathcal{L}^{(p)}_{I_o}(s)$.
Zhang \emph{et al}. \cite{Zhang2013Enhancing}, \cite{Louie2011Open} derived a closed-form expression of $\mathcal{L}^{(p)}_{I_o}(s)$ for ad hoc networks.
 In an ad hoc network, interfering nodes can be arbitrarily close to the desired receiver, which however is not possible in a cellular network.
As mentioned above, with mobile association in cellular networks, interfering BSs are always located outside a certain region around the desired UE.
For the cellular model, Li \emph{et al.} \cite{Li2014Throughput} proposed a useful approach to handle $\mathcal{L}^{(p)}_{I_o}(s)$ by first expressing it in a recursive form, and then transforming it into a lower triangular Toeplitz matrix form.
This yields a tractable expression of the connection probability for multi-antenna transmissions in cellular networks.

The analysis of the PLS in HCNs is quite different from \cite{Li2014Throughput} where only a single-tier network without secrecy demands is considered.
In addition, due to the difficulty in deriving the distribution of the interference signal along with artificial noise, the derivation of $\mathcal{L}^{(p)}_{I_o}(s)$ becomes much more complicated.
Fortunately, we provide an accurate integral form of $\mathcal{P}_{c,k}$ in the following theorem.

\begin{theorem}\label{pck_theorem}
\textit{The connection probability of a typical UE associated with tier $k$ is given by
\begin{align}\label{pc}
     \mathcal{P}_{c,k}=
     \frac{\pi\lambda_k}{\mathcal{S}_k}&\sum_{i=0}^{M_k-1}
\sum_{m=0}^{M_k-1}\sum_{p=0}^{m}\nonumber\\
     &\left(\frac{\beta_{t}N_0}{\phi_kP_k}\right)^{m-p}
     \frac{\mathcal{Z}_{k,m,p,i}}{i!(m-p)!}
     \mathbf{\Theta}_{M_k}^i(p+1,1),
\end{align}
where $\mathcal{Z}_{k,m,p,i}=\int_0^{D_k^2}
x^{i+\frac{\alpha}{2}(m-p)}e^{-\frac{\beta_{t}N_0}{\phi_kP_k}  x^{\frac{\alpha}{2}}-\pi\Upsilon_k x}dx$, and \\
      $\Upsilon_k=\sum_{j\in\mathcal{K}}\lambda_j
      \left({{P}_{j,k}{M}_{j,k}}
        \right)^{{\delta}}
      \Big\{1-\mathcal{A}_j
      +\left(\frac{\phi_{j,k}\beta_{t}}{{M}_{j,k}}
        \right)^{{\delta}}\mathcal{A}_j\Upsilon_{j1}
        +\frac{{\delta}{M}_{j,k}}{\phi_{j,k}\beta_{t}}
        \mathcal{A}_j\Upsilon_{j2}\Big\}$,
  with
\begin{align}\label{G1}
  \Upsilon_{j1} = \begin{cases}
  ~C_{\alpha,M_j+1} , &\xi_j=1,\\
  ~\frac{C_{\alpha,2} }{(1-\xi_j)^{M_j-1}} -\sum\limits_{n=0}^{M_j-2}
  \frac{\xi_j^{1+{\delta}}C_{\alpha,n+2}}
{(1-\xi_j)^{M_j-1-n}}, &\xi_j\neq1,
  \end{cases}
\end{align}
and $\Upsilon_{j2}$ shown in \eqref{G2} at the top of this page.
${_2}F_1(\cdot)$ denotes the Gauss hypergeometric function,
\begin{figure*}[t]
\begin{align}\label{G2}
\Upsilon_{j2}= \begin{cases}
  ~\left(\frac{{M}_{j,k}}{\phi_{j,k}\beta_{t}}\right)^{M_j-1}
    \frac{{_2}F_1\left(M_j,M_j+\delta;M_j+\delta+1;-\frac{ M_{j,k}}{\phi_{j,k}\beta_{t}}\right)}
    {M_j+{\delta}}
, &\xi_j=1,\\
  ~{\frac{{_2}F_1\left(1,{\delta}+1;{\delta}+2;
    -\frac{M_{j,k}}{\phi_{j,k}\beta_{t}}\right)}
  {1+{\delta}(1-\xi_j)^{M_j-1}}
  -\sum\limits_{n=0}^{M_j-2}
  \left( \frac{ M_{j,k}}
  {\xi_j\phi_{j,k}\beta_{t}}\right)^n\frac{{_2}F_1\left(
    n+1,n+1+{\delta};n+2+{\delta};
    -\frac{M_{j,k}}{\xi_j\phi_{j,k}\beta_{t}}\right)}
    {\left(n+1+{\delta}\right)\left(1-\xi_j\right)^{M_j-1-n}}},
  &\xi_j\neq1,
  \end{cases}
\end{align}
\hrulefill
\end{figure*}
and $\mathbf{\Theta}_M^i(p,q)$ denotes the row-$p$-column-$q$ entry of $\mathbf{\Theta}_M^i$, where $\mathbf{\Theta}_M$ is a Toeplitz matrix
\begin{align}
\mathbf{\Theta}_M \triangleq \left[\begin{array}{ccc}
    0 & & \\
    g_1 & 0  & \\
    g_2 & g_1 & 0  ~~~ ~~~  \\
    \vdots  & & ~~\ddots\\
    g_{M-1} & g_{M-2} & \cdots ~~g_1 ~~ 0
\end{array}\right],
\label{QM}
\end{align}
with $g_{i}=\frac{\pi{\delta}}{i-{\delta}}
\sum_{j\in\mathcal{K}}\mathcal{A}_j\lambda_j
\left({ P_{j,k} M_{j,k}}\right)^{{\delta}}
    \left(\frac{\phi_{j,k}\beta_{t}}{M_{j,k}}\right)^{i}
\mathcal{Q}_{j,i}$ where $\mathcal{Q}_{j,i}$ is shown in \eqref{Q} at the top of this page.
\begin{figure*}[t]
\begin{align}\label{Q}
    \mathcal{Q}_{j,i}=\begin{cases}
    ~\binom{M_j+i-1}{M_j-1}
  {_2}F_1\left(M_j+i,i-{\delta};i-{\delta}+1;
    -\frac{\phi_{j,k}\beta_{t}}{M_{j,k}}\right), &\xi_j=1,\\
    ~\frac{{_2}F_1\left(i+1,i-{\delta};i-{\delta}+1;
    -\frac{\phi_{j,k}\beta_{t}}{M_{j,k}}\right)
    -\sum\limits_{n=0}^{M_j-2}\binom{n+i}{n}
    \frac{\xi_j^{i+1}}{(1-\xi_j)^{n}}
    {{_2}F_1\left(n+i+1,i-{\delta};i-{\delta}+1;
    -\frac{\xi_j\phi_{j,k}\beta_{t}}{M_{j,k}}\right)}
    }{(1-\xi_j)^{M_j-1}}, & \xi_j\neq1,
    \end{cases}
\end{align}
\hrulefill
\end{figure*}}
\end{theorem}
\begin{proof}
    Please see Appendix \ref{pck_theorem_proof}.
\end{proof}

Although \eqref{pc} seems to be rather unwieldy due to the integral term $\mathcal{Z}_{k,m,p,i}$ and ${_2}F_1(\cdot)$, it is actually easy to compute.
Theorem \ref{pck_theorem} provides a general and accurate expression for the connection probability without requiring time-consuming Monte Carlo simulations.
 It also provides a baseline for comparison with other approximate results.
It appears impossible to extract main properties of the connection probability from \eqref{pc}, thus motivating the need for more compact forms.

\subsection{Interference-limited HCN}
Due to ubiquitous interference in the HCN, the interference at a receiver apparently dominates thermal noise.
Therefore, it is reasonable for us to consider the interference-limited case by ignoring thermal noise, i.e., $N_0=0$.
In this case, the connection probability \eqref{pc_k} simplifies to
\begin{equation}\label{pc_k2}
  \mathcal{P}^{int}_{c,k}=
  \sum_{m=0}^{M_k-1}\mathbb{E}_{R_k}\left[
  \frac{(-1)^m s^m}
    {m!}
    \mathcal{L}^{(m)}_{I_o}(s)\right],
\end{equation}
with a more analytically tractable expression of $\mathcal{P}^{int}_{c,k}$ provided by the following corollary.

\begin{corollary}\label{pck_corollary}
\textit{For the interference-limited HCN, the connection probability of a UE associated with tier $k$ is 
\begin{align}\label{pc_inter}
     \mathcal{P}^{int}_{c,k}&=\frac{\lambda_k}{\mathcal{S}_k}
     \sum_{m=0}^{M_k-1}
     \frac{\left\|\mathbf{\Theta}_{M_k}^m\right\|_1}
     {\pi^{m}\Upsilon^{m+1}_k}
     \left(1-\sum_{l=0}^m\frac{\pi^{l}e^{-\pi\Upsilon_k D_k^2}}{l!D_k^{-2l}\Upsilon^{-l}_k}\right),
\end{align}
where $\|\cdot\|_1$ is the $L_1$ induced matrix norm, i.e., $\|\mathbf{A}\|_1=\max_{1\le j\le N}\sum_{i=1}^M|A_{ij}|$ for $\mathbf{A}\in\mathbb{R}^{M\times N}$.}
\end{corollary}
\begin{proof}
    Please see Appendix \ref{pck_corollary_proof}.
\end{proof}

Corollary \ref{pck_corollary} provides a much simpler expression for the connection probability than \eqref{pc}.
Note that the term $\sum_{l=0}^m\frac{\pi^{l}e^{-\pi\Upsilon_k D_k^2}}
     {l!D_k^{-2l}\Upsilon^{-l}_k}$ in \eqref{pc_inter} is a consequence of $\tau\neq 0$.
     This term goes to zero as $\tau\rightarrow 0$ (i.e., non-threshold mobile association) since $\lim_{\tau\rightarrow 0}D_k\rightarrow\infty$.
As shown in Fig. \ref{PC_P1_TAU_COM}, $\mathcal{P}_{c,k}$ and $\mathcal{P}^{int}_{c,k}$  nearly merge, and in the subsequent analysis we focus on the latter for convenience.

\begin{figure}[!t]
\centering
\includegraphics[width=3.0in]{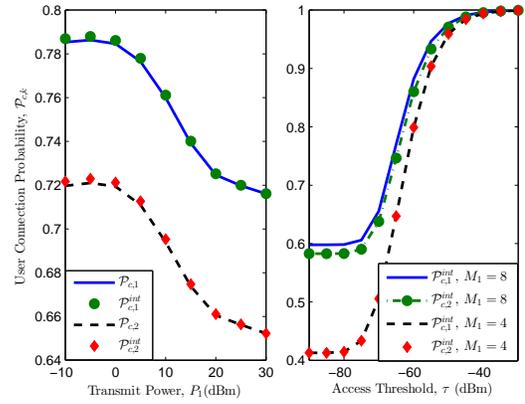}
\caption{Connection probabilities in a 2-tier HCN.
 $\mathcal{P}_{c,k}$ in \eqref{pc} and $\mathcal{P}^{int}_{c,k}$ in \eqref{pc_inter} vs. $P_1$ in the left figure and
 $\mathcal{P}^{int}_{c,k}$ vs. $\tau$ for different $M_1$'s in the right figure, with $\{P_1,P_2\}=\{30,10\}$dBm, $M_2=4$, $\{\lambda_2,\lambda_u\}=\{4, 10\}\lambda_1$, $\beta_{t}=10$, and $\{\phi_1,\phi_2\}=\{0.8, 0.8\}$.
 Unless otherwise specified, we set $\alpha=4$ and $\lambda_1=\frac{1}{\pi 400^2\mathrm{m}^2 }$.}
\label{PC_P1_TAU_COM}
\end{figure}

\subsection{Asymptotic Analysis on $\mathcal{P}^{int}_{c,k}$}
It is important for network design to understand how $\mathcal{P}^{int}_{c,k}$ is affected by system parameters such as $\tau$, $P_k$, $\lambda_k$ and $\lambda_u$, etc.
Since $\Upsilon_k$ and $\mathbf{\Theta}_{M_k}$ in \eqref{pc_inter} are coupled through these parameters in a very complicated way, the relationship between them is generally not explicit.
In the following, we provide some insights into the behavior of $\mathcal{P}^{int}_{c,k}$ w.r.t. the above parameters by performing an asymptotic analysis, with the corresponding proof relegated to Appendix \ref{proof_property_2_5}.
\begin{property}\label{pc_property2}
\textit{For the case that all tiers share the same number of BS antennas $M$ and power allocation ratio $\phi$, and $\lambda_u\gg \lambda_j$, $\forall j\in\mathcal{K}$, $\mathcal{P}^{int}_{c,k}$ converges to a value that is independent of the transmit power $P_j$, BS density $\lambda_j$ and $k\in\mathcal{K}$ as $\tau\rightarrow0$.}
\end{property}

\begin{property}\label{pc_property3}
\textit{$\mathcal{P}_{c,k}^{int} \rightarrow 1$ as $\tau\rightarrow\infty$ for $k\in\mathcal{K}$.
}
\end{property}

\begin{property}\label{pc_property4}
\textit{When the transmit power of tier 1 is much larger than that of the other tiers, $\mathcal{P}^{int}_{c,k}$ increases with $\tau$ and $\lambda_l$, $\forall l\neq1$, and decreases with $\lambda_u$, $\forall k\in\mathcal{K}$.}
\end{property}

\begin{property}\label{pc_property5}
   \textit{When $P_{j,1}\ll 1$, $\forall j\neq 1$, $\mathcal{P}^{int}_{c,k}$ decreases with $P_1$, and converges to a constant value as $P_1\rightarrow \infty$, $\forall k\in\mathcal{K}$.
   }
\end{property}

Property \ref{pc_property2} shows that under a loose control on mobile access ($\tau\rightarrow0$) with $\{M_j\}=M$ and $\{\phi_j\}=\phi$, the connection probability becomes insensitive to transmit power and BS densities, i.e., increasing transmit power or randomly adding new infrastructure
does not influence connection performance (link quality).
 This \emph{insensitivity} property obtained for this special case is also observed in a single-antenna unbiased HCN \cite{Dhillon2012Modeling,Jo2012Heterogeneous}.

Property \ref{pc_property3} implies that increasing $\tau$ increases connection probability, just as explained in Sec. II-D.
Nevertheless, $\tau$ should not be set as large as possible in practice.
As will be observed later in Sec. VI, $\tau$ should be properly chosen to achieve a good secrecy throughput performance under certain connection constraints.

Properties \ref{pc_property2} and \ref{pc_property3} are validated in Fig. \ref{PC_P1_TAU_COM}.
We find that both $\mathcal{P}_{c,1}^{int}$ and $\mathcal{P}_{c,2}^{int}$ increase with $M_1$.
The reason is that, on one hand a larger $M_1$ produces a higher diversity gain, and improves the link quality for tier 1.
On the other hand, a larger $M_1$ also provides a stronger bias towards admitting UEs, thus the UE originally associated with tier 2 under low link quality (e.g., at the edge of a cell in tier 2) now connects to tier 1, which as a consequence enhances the link quality of tier 2.

\begin{figure}[!t]
\centering
\includegraphics[width=3.0in]{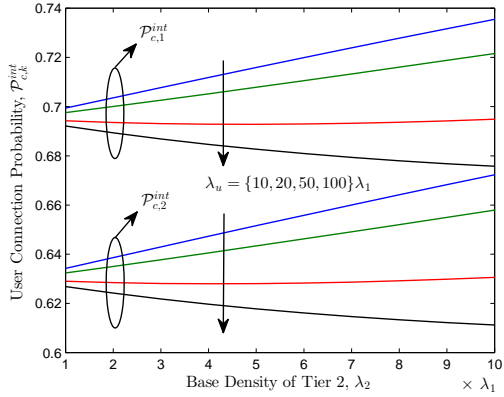}
\caption{Connection probability in a 2-tier HCN vs. $\lambda_2$ for different $\lambda_u$'s, with $\alpha=4$, $\{P_1,P_2\}=\{30,10\}$dBm, $\{M_1,M_2\}=\{6,4\}$, $\lambda_1=\frac{1}{\pi 400^2\mathrm{m}^2 }$, $\tau=-90$dBm, $\beta_{t}=5$, and $\{\phi_1,\phi_2\}=\{1, 0.5\}$.}
\label{Pc_La_Lu}
\end{figure}
Property \ref{pc_property4} provides some interesting counter-intuitive insights into connection performance.
For instance, deploying more pico/femto BSs may improve connection probabilities.
This is because a larger $\lambda_l$ decreases the number of active BSs in the other tiers, which reduces the aggregate network interference especially when the transmit power of the other tiers is large.
However, the connection probability decreases when more UEs are introduced, since more BSs are now activated, resulting in greater interference.
Although Property \ref{pc_property4} is obtained as ${P_{j,1}}\rightarrow 0$, it applies more generally, as illustrated in Fig. \ref{Pc_La_Lu}.
We see that, a larger $\lambda_u$ decreases $\mathcal{P}_{c,k}^{int}$ for $k=1,2$.
In addition, we observe that, when $\frac{\lambda_u}{\lambda_2}\le 50$, a larger $\lambda_2$ increases $\mathcal{P}_{c,k}^{int}$, whereas when $\frac{\lambda_u}{\lambda_2}>50$, it decreases $\mathcal{P}_{c,k}^{int}$.
The underlying reason is that, in the latter both $\mathcal{A}_1$ and $\mathcal{A}_2$ nearly reach one, hence deploying more microcells significantly increases interference, which deteriorates link reliability.
Nevertheless, this performance degradation can be effectively mitigated by setting a larger access threshold, since in this way more BSs remain idle, alleviating the network interference.

\begin{figure}[!t]
\centering
\includegraphics[width=3.0in]{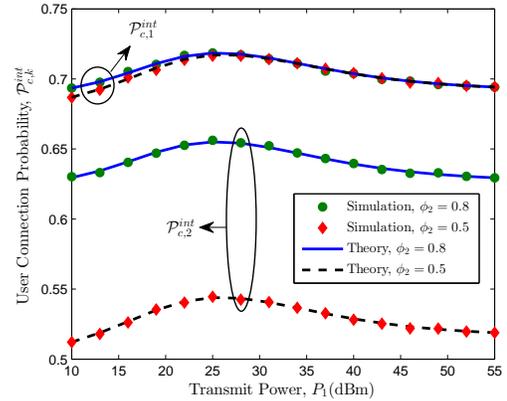}
\caption{Connection probability in a 2-tier HCN vs. $P_1$ for different $\phi_2$'s, with $\alpha=4$, $P_2=\mathrm{10}$dBm, $\{M_1,M_2\}=\{6,4\}$, $\lambda_1=\frac{1}{\pi 400^2\mathrm{m}^2 }$, $\{\lambda_2,\lambda_u\}=\{2\lambda_1, 5\lambda_1\}$, $\tau = -90$dBm, $\beta_{t}=5$, and $\phi_1=1$.}
\label{Pc_P_Phi}
\end{figure}

Property \ref{pc_property5} implies that as $P_1$ increases, $\mathcal{P}^{int}_{c,1}$ first increases and then decreases, and eventually levels off.
Increasing transmit power is not always beneficial to connection performance, since the growth of signal power is counter-balanced by the growth of interference power.
The same is true for $\mathcal{P}^{int}_{c,j}$, $\forall j\neq 1$.
The underlying reason is that, as $P_1$ gets larger, $\mathcal{A}_1$ increases while $\mathcal{A}_j$ decreases, which increases interference from tier $1$ while reducing that from the other $K-1$ tiers.
When $P_1$ is relatively small, the decrement of interference from the other $K-1$ tiers dominates the increment of interference from tier $1$, so the aggregate interference actually reduces;
the opposite occurs as $P_1$ further increases.
This property is confirmed in Fig. \ref{Pc_P_Phi}.
We also find that, for a given $P_1$, $\mathcal{P}^{int}_{c,2}$ increases significantly with $\phi_2$, while $\mathcal{P}^{int}_{c,1}$ experiences negligible impact since the inference in tier 1 varies little.
Even though, we still observe a slight improvement in $\mathcal{P}^{int}_{c,1}$ in the small $P_1$ region as $\phi_2$ increases.
This occurs because focusing more power on the desired UE in tier 2 to some degree decreases the residual interference (artificial noise and leaked signal) to the UE associated with tier 1.
 However, the reduced interference becomes negligible as $P_1$ increases, and $\mathcal{P}^{int}_{c,1}$ becomes insensitive to $\phi_2$.

\section{User Secrecy Probability}
In this section, we investigate the secrecy probability of a randomly located UE.
 Here secrecy probability corresponds to the probability that a secret message is not decoded by any Eve.

When accounting for secrecy demands, we should not underestimate the wiretap capability of Eves.
In our work, we consider a worst-case wiretap scenario in which Eves are assumed to have the capability of multiuser decoding (e.g., successive interference cancellation), thus the interference created by concurrent transmission of information signals can be completely resolved \cite{Zhang2013Enhancing}.
Therefore, Eves only receive the artificial noise from all BSs in the HCN.
The received signal at $\textsf{E}_e$ is given by
\begin{align}\label{ye}
    y_{e} =& \underbrace{\frac{\sqrt{\phi_k P_k}\mathbf{h}_{be}^\mathrm{T}\mathbf{w}_{b}s_{b}}
    {r_{be}^{\alpha/2}}}_{\mathrm{information~signal}}
    +\underbrace{\frac{\sqrt{(1-\phi_k) P_k}\mathbf{h}_{be}^{\mathrm{T}}
    \mathbf{W}_{b}\mathbf{v}_{b}}{r_{be}^{\alpha/2}}}
    _{\mathrm{serving-BS~artificial~noise}}\nonumber\\
    &+ \underbrace{\sum_{j\in\mathcal{K}}\sum_{z\in\Phi_j^o
    \setminus b} \frac{\sqrt{(1-\phi_j) P_j}\mathbf{h}_{ze}^{\mathrm{T}}
    \mathbf{W}_{z}\mathbf{v}_{z}}{r_{ze}^{\alpha/2}}}
    _{\mathrm{intra-~ and ~ cross-tier~artificial~noise}}
    +n_{e}.
\end{align}

We consider the non-colluding wiretap scenario where each Eve individually decodes secret messages.
In that case, transmission is secure only if secrecy is achieved against all Eves.
Accordingly, the secrecy probability of tier $k$ is defined as the probability of the
event that the instantaneous SINR of an arbitrary Eve falls below a target SINR $\beta_{e}$, i.e.,
\begin{align}\label{ps_k}
    \mathcal{P}_{s,k} \triangleq &\mathbb{E}_{\Phi_1}\cdots
    \mathbb{E}_{\Phi_K}\mathbb{E}_{\Phi_e}\nonumber\\
   & \left[\prod_{e\in\Phi_e}\mathbb{P}
    \left\{\mathbf{SINR}_{e,k}
    <\beta_{e}|\Phi_e, \Phi_1,\cdots,\Phi_K\right\}\right],
\end{align}
where $\mathbf{SINR}_{e,k}$ is given by
\begin{equation}\label{sinr_e}
    \mathbf{SINR}_{e,k} = \frac{\phi_k P_k |\mathbf{h}_{be}^\mathrm{T}\mathbf{w}_{b}|^2 r^{-\alpha}_{be}}
    {I_{be}+\sum_{j\in\mathcal{K}}I_{je}+N_0},~\forall e\in\Phi_e,
\end{equation}
with $I_{be}\triangleq\frac{(1-\phi_k)P_k
|\mathbf{h}_{be}^{\mathrm{T}}
\mathbf{W}_{b}|^2}{(M_k-1)r^{\alpha}_{be}}$
and $I_{je}\triangleq\sum_{z\in\Phi^o_j\setminus b}\frac{(1-\phi_j)P_j\|\mathbf{h}_{ze}^{\mathrm{T}}
\mathbf{W}_{z}\|^2}{(M_j-1)r^{\alpha}_{ze}}$.
Unfortunately, it is intractable to derive an exact expression for $\mathcal{P}_{s,k}$ from \eqref{ps_k}.
Instead, we provide the upper and lower bounds of $\mathcal{P}_{s,k}$ in the following theorem as done in \cite{Zhou2011Throughput} and \cite{Zhang2013Enhancing}.

\begin{theorem}  \label{psk_theorem}
\textit{The secrecy probability $\mathcal{P}_{s,k}$ in \eqref{ps_k} satisfies
\begin{equation}\label{ps_k_inequality}
  \mathcal{P}_{s,k}^L \le \mathcal{P}_{s,k} \le \mathcal{P}_{s,k}^U,
\end{equation}
where
\begin{align}\label{ps_k_lower}
    \mathcal{P}_{s,k}^L &= \exp\left(-\frac{\pi\lambda_e}
    {\left(1+\xi_k\beta_{e}\right)^{M_k-1}}
    \int_0^{\infty}e^{-\frac{\beta_{e} N_0}{\phi_kP_k}r^{\frac{\alpha}{2}}
    -\pi\psi_k{\beta_e^\delta} r}dr\right),\\
\label{ps_k_upper}
    \mathcal{P}_{s,k}^U &= 1-\frac{\pi\lambda_e}
    {\left(1+\xi_k\beta_{e}\right)^{M_k-1}}
    \int_0^{\infty}e^{-\frac{\beta_{e} N_0}{\phi_kP_k}r^{\frac{\alpha}{2}}
    -\pi\psi_k{\beta_e^\delta}r-\pi\lambda_e r}dr,
\end{align}
with $\psi_k\triangleq\sum_{j\in\mathcal{K}}
    \mathcal{A}_j\lambda_jC_{\alpha,M_j}
    \left({\xi_j\phi_{j,k} P_{j,k}}\right)^
  {{\delta}}$.}
  \end{theorem}
\begin{proof}
    Please see Appendix \ref{psk_theorem_proof}.
\end{proof}

Interestingly, when $\lambda_e\ll 1$, the two bounds merge, i.e.,
\begin{equation}\label{ps_k_bound}
    \mathcal{P}_{s,k}^L \approx \mathcal{P}_{s,k}^{o}:=1-\frac{\pi\lambda_e\int_0^{\infty}e^{-\frac{\beta_{e} N_0}{\phi_kP_k}r^{\frac{\alpha}{2}}
    -\pi\psi_k{\beta_e^\delta}r}dr}
    {\left(1+\xi_k\beta_{r,k}\right)^{M_k-1}}
    \approx \mathcal{P}_{s,k}^U,
\end{equation}
and $ \mathcal{P}_{s,k}^{o}\rightarrow 1$ as $\lambda_e\rightarrow 0$.
 It implies that as $\mathcal{P}_{s,k}\rightarrow 1$, both the upper and lower bounds approach the exact value.
 In other words, $\mathcal{P}_{s,k}$ can be approximated by $\mathcal{P}_{s,k}^{o}$ in the high secrecy probability region.

For some special cases, the calculation for the upper and lower bounds can be simplified.
For example, substituting $\alpha=4$ into Theorem \ref{psk_theorem}, we have the following corollary.
\begin{corollary}\label{psk_4}
\textit{When $\alpha = 4$, the secrecy probability of tier $k$ $\mathcal{P}^{\alpha=4}_{s,k}$ satisfies
\begin{equation}\label{ps_k_ineuq_1}
  \mathcal{P}_{s,k}^{L,\alpha = 4} \le \mathcal{P}^{\alpha = 4}_{s,k} \le \mathcal{P}_{s,k}^{U,\alpha = 4},
\end{equation}
where
\begin{align}\label{ps_k_lower_1}
    \mathcal{P}_{so}^{L,\alpha = 4} &= \exp\left(-\frac{\pi^{\frac{1}{2}}\lambda_e
    e^{\frac{\psi_k^2
    {\beta_e^{2\delta}}}
    {\gamma_k}}}{\sqrt{\gamma_k}\left(1+\xi_k\beta_{e}\right)^{M_k-1}}
    \left(1-\Omega\left(\frac{\psi_k{\beta_e^\delta}}{\sqrt{\gamma_k}}
\right)\right)\right),\\
\label{ps_k_upper_1}
    \mathcal{P}_{s,k}^{U,\alpha = 4} &= 1-\frac{\pi^{\frac{1}{2}}\lambda_e
    e^{\frac{ (\lambda_e+\psi_k{\beta_e^\delta})^2}
    {\gamma_k}}}  {\sqrt{\gamma_k}\left(1+\xi_k\beta_{e}\right)^{M_k-1}}
    \left(1-\Omega\left(\frac{\lambda_e+\psi_k{\beta_e^\delta}}
    {\sqrt{\gamma_k}}
\right)\right),
\end{align}
with $\Omega(x)\triangleq\frac{1}{\sqrt{\pi}}\int_0^{x^2}
\frac{e^{-t}}{\sqrt{t}}dt$ and $\gamma_k\triangleq\frac{4\beta_{e}N_0}{\pi^2\phi_kP_k}$.}
\end{corollary}
From Corollary \ref{psk_4}, an approximation for $\mathcal{P}_{s,k}^{\alpha = 4}$ in the high secrecy probability region is
\begin{equation}\label{ps_k_bound_1} \mathcal{P}_{s,k}^{o,\alpha := 4}=1-\frac{\pi^{\frac{1}{2}}\lambda_e
    e^{\frac{\psi_k^2\beta_e^{2\delta}}
    {\gamma_k}}}  {\sqrt{\gamma_k}\left(1+\xi_k\beta_{e}\right)^{M_k-1}}
    \left(1-\Omega\left(\frac{\psi_k\beta_e^{\delta}}
    {\sqrt{\lambda_k} } \right)\right).
\end{equation}
Substituting $N_0=0$ into Theorem \ref{psk_theorem}, we have the following corollary.
\begin{corollary}\label{psk_inter_corollary}
\textit{For an interference-limited HCN, the secrecy probability of tier $k$ $\mathcal{P}^{int}_{s,k}$ satisfies
\begin{equation}\label{ps_ineuq_2}
  \mathcal{P}_{s,k}^{int,L} \le \mathcal{P}^{int}_{s,k} \le \mathcal{P}_{s,k}^{int,U},
\end{equation}
where
\begin{align}\label{ps_k_lower_2}
    \mathcal{P}_{s,k}^{int,L} &= \exp\left(-{\lambda_e\psi_k^{-1}\beta_e^{-\delta}}
    {\left(1+\xi_k\beta_{e}\right)^{1-M_k}}\right),\\
\label{ps_k_upper_2}
    \mathcal{P}_{s,k}^{int,U} &= 1-\frac{\lambda_{e}}
    {\lambda_e+\psi_k\beta_e^{\delta}}
    \left(1+\xi_k\beta_{e}\right)^{1-M_k}.
\end{align}}
\end{corollary}

Corollary \ref{psk_inter_corollary} provides an approximation for $\mathcal{P}_{s,k}^{int}$ in the high secrecy probability region as
\begin{equation}\label{ps_k_bound_2} \mathcal{P}_{s,k}^{int,o}:=1-{\lambda_e\psi_k^{-1}\beta_e^{-\delta}}
    {\left(1+\xi_k\beta_{e}\right)^{1-M_k}}.
\end{equation}
 Next we establish some properties on $\mathcal{P}_{s,k}^{int,o}$, with their corresponding proofs relegated to Appendix \ref{proof_property_6_8}.
\begin{property}\label{ps_property6}
    \textit{$\mathcal{P}_{s,k}^{int,o}$ monotonically decreases in $\lambda_e$, $\tau$, and $\phi_j$, $\forall j\in\mathcal{K}$, and it increases in $\lambda_u$.}
\end{property}

\begin{property}\label{ps_property7}
    \textit{When $M_{j,k}\ll 1$, $\forall j\in\mathcal{K}\setminus k$, $\mathcal{P}_{s,k}^{int,o}$ increases in $M_k$, and decreases in $\lambda_k$.}
\end{property}

\begin{property}\label{ps_property8}
\textit{In the high $P_k$ region, $\mathcal{P}^{int,o}_{s,j}$, $\forall j\in\mathcal{K}$, increases in $P_k$;
$\mathcal{P}^{int,o}_{s,k}$ converges to a constant value as $P_k\rightarrow\infty$ and  $\lim_{P_k\rightarrow\infty}\mathcal{P}^{int,o}_{s,l}= 1$, $\forall l\neq k$.}
\end{property}

 Properties \ref{ps_property6}-\ref{ps_property8} provide insights into the secrecy probability that differ from those obtained about the connection probability.
For example, deploying more pico/femto BSs may increase  connection probabilities while reducing the secrecy probability, which implies that proper BS densities should be designed to balance link quality and secrecy.
 The above properties are further validated by the following numerical examples.

\begin{figure}[!t]
\centering
\includegraphics[width=3.0in]{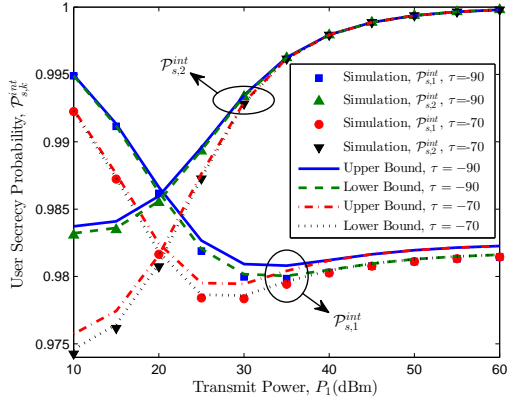}
\caption{Secrecy probability in Corollary \ref{psk_inter_corollary} in a 2-tier HCN vs. $P_1$ for different $\tau$(dBm)'s, with $\alpha=4$, $P_2=20$dBm, $\{M_1,M_2\}=\{6,4\}$, $\lambda_1=\frac{1}{\pi 400^2\mathrm{m}^2  }$,  $\{\lambda_2,\lambda_u,\lambda_e\}=
\{2,2,0.05\}\lambda_1$, $\beta_{e}=1$, and $\{\phi_1,\phi_2\}=\{0.5, 0.5\}$.}
\label{Ps_P_Tau}
\end{figure}

Fig. \ref{Ps_P_Tau} depicts secrecy probability versus $P_1$ for different values of $\tau$.
We see that, the lower bound accurately approximates the simulated value, while the upper bound becomes asymptotically tight in the high secrecy probability region.
As $P_1$ increases, $\mathcal{P}^{int}_{s,1}$ first decreases and then slowly rises to a constant value that is independent of $P_1$.
$\mathcal{P}^{int}_{s,2}$ reaches one as $P_k$ becomes large enough, verifying Property \ref{ps_property8}.
We observe that, the secrecy probabilities of both tiers increase as $\tau$ decreases, while Table I shows that connection probabilities decrease as $\tau$ decreases, as indicated in Property \ref{pc_property2}.
The access threshold $\tau$ displays a tradeoff between the connection and secrecy probabilities.
This is because a smaller $\tau$ results in more interference, which simultaneously degrades the legitimate and wiretap channels.
In other words, network interference is a double-edged sword that promotes the secrecy transmission but in turn restrains the legitimate communication.

\begin{table}
 \caption{$\mathrm{Connection~Probability~ vs.~ Secrecy~ Probability }$}
\begin{center}
  \begin{tabular}{|c|c|c|c|c|c|c|c|}
  \hline
    Probabilities &  $\mathcal{P}^{int}_{c,1}$ & $\mathcal{P}^{int}_{c,2}$ & $\mathcal{P}^{int}_{s,1}$ &  $\mathcal{P}^{int}_{s,2}$ \\\hline
    $\tau = -70$dBm &  0.9571 & 0.9477 &0.9786 &  0.9930   \\\hline
    $\tau = -90$dBm &  0.9186& 0.9073 &0.9799& 0.9936 \\
  \hline
\end{tabular}
\end{center}
\label{Pc_Ps}
\end{table}

\begin{figure}[!t]
\centering
\includegraphics[width=3.0in]{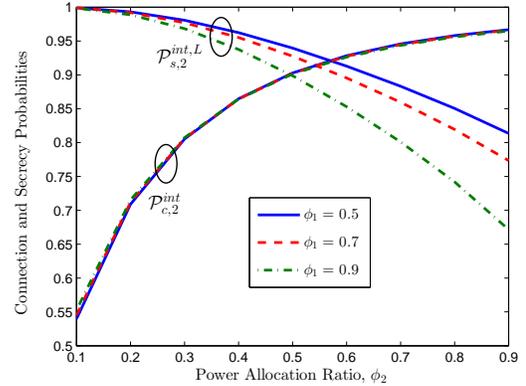}
\caption{Connection probability in \eqref{pc_inter} and secrecy probability in \eqref{ps_k_lower_2} in a 2-tier HCN vs. $\phi_2$ for different $\phi_1$'s, with $\alpha=4$, $\{P_1,P_2\}=\{30,20\}$dBm, $\{M_1,M_2\}=\{6,4\}$, $\lambda_1=\frac{1}{\pi 400^2 \mathrm{m}^2 }$,  $\{\lambda_2,\lambda_u,\lambda_e\}=
\{2,4,0.5\}\lambda_1$, $\beta_{t}=2$, $\beta_{e}=1$, and $\tau = -90$dBm.}
\label{Pc_Ps_Phi}
\end{figure}

In Fig. \ref{Pc_Ps_Phi}, we see that as $\phi_2$ increases, i.e., more power is allocated to the information signal, $\mathcal{P}^{int}_{c,2}$ increases and $\mathcal{P}^{int}_{s,2}$ decreases.
We should design the power allocation to strike a better balance between reliability and secrecy.
In addition, although a smaller $\phi_1$ rarely affects $\mathcal{P}^{int}_{c,2}$, it significantly increases $\mathcal{P}^{int}_{s,2}$, which highlights the validity of the artificial-noise method.

\section{Network-wide Secrecy Throughput}
In this section, we investigate the network-wide secrecy throughput of the HCN, in terms of the \emph{secrecy transmission capacity} \cite{Zhou2011Throughput}, \cite{Zhang2013Enhancing}, which is defined as the achievable rate of successful transmission of secret messages per unit area subject to both connection and secrecy probability constraints.
Mathematically, the network-wide secrecy throughput, with a connection probability constraint $\mathcal{P}_{c,k}(\beta_{t,k})=\varrho$ and a secrecy probability constraint $\mathcal{P}_{s,k}(\beta_{e,k})=\epsilon$ for $k\in\mathcal{K}$, is given by
\begin{align}\label{st}
\mathcal{T} &
=\sum_{k\in\mathcal{K}}
\lambda_k\mathcal{A}_k\varrho\mathcal{R}^*_{s,k}
=\sum_{k\in\mathcal{K}}
\lambda_k\mathcal{A}_k\varrho
    \left[\mathcal{R}^*_{t,k} - \mathcal{R}^*_{e,k}\right]^{+}\nonumber\\
&= \sum_{k\in\mathcal{K}}
\lambda_k\mathcal{A}_k\varrho
     \left[\log\left(\frac{1+\beta^*_{t,k}}
    {1+\beta^*_{e,k}}\right)\right]^{+},
\end{align}
where $\mathcal{R}^*_{s,k}=\left[\mathcal{R}^*_{t,k} - \mathcal{R}^*_{e,k}\right]^{+}$, $\mathcal{R}^*_{t,k}=\log (1+\beta^*_{t,k})$ and $\mathcal{R}^*_{e,k}= \log (1+\beta^*_{e,k})$ are the secrecy, codeword and redundant rates for tier $k$, with
$\beta^*_{t,k}$ and $\beta^*_{e,k}$ the unique roots of the equations $\mathcal{P}_{c,k}(\beta_{t,k})=\varrho$
and $\mathcal{P}_{s,k}(\beta_{e,k})=\epsilon$, respectively.
Note that, if $\mathcal{R}^*_{t,k} - \mathcal{R}^*_{e,k}$ is negative, the connection and secrecy probability constraints can not be satisfied simultaneously, and transmissions should be suspended.
When system designers can control the connection and secrecy probability constraints, \eqref{st} can be used to find the optimal values of $\varrho$ and $\epsilon$ to maximize $\mathcal{T}$.
    In the following, we calculate $\beta^*_{t,k}$ and $\beta^*_{e,k}$ from $\mathcal{P}_{c,k}(\beta_{t,k})=\varrho$
and $\mathcal{P}_{s,k}(\beta_{e,k})=\epsilon$, respectively.
For convenience, we consider an interference-limited HCN.

Due to the complicated expression of $\mathcal{P}_{c,k}^{int}$ in \eqref{pc_inter}, we are unable to derive an expression for $\beta^*_{t,k}$ from $\mathcal{P}^{int}_{c,k}(\beta_{t,k})=\varrho$.
However, since $\mathcal{P}^{int}_{c,k}(\beta_{t,k})$ is obviously a monotonically decreasing function of $\beta_{t,k}$, we can efficiently calculate $\beta^*_{t,k}$ that satisfies $\mathcal{P}^{int}_{c,k}(\beta_{t,k})=\varrho$ using the bisection method.

To guarantee a high level of secrecy, the secrecy probability $\epsilon$ must be large, which allows us to use \eqref{ps_k_bound_2} to calculate $\beta^*_{e,k}$.
For the case $M_k\ge 3$, we can only numerically calculate $\beta^*_{e,k}$ that satisfies $\mathcal{P}^{int}_{s,k}(\beta_{e,k})=\epsilon$ using the bisection method.
Fortunately, when $M_k=2$ or $M_k\gg1$, we can provide closed-form expressions of $\mathcal{\beta}^*_{e,k}$, with corresponding proof relegated to Appendix \ref{proof_proposition_1_2}.

\begin{proposition}\label{beta_e_proposition1}
\textit{In the large $\epsilon$ regime with $\alpha=4$ and $M_k=2$, the root of $\mathcal{P}^{int}_{s,k}(\beta_{e,k})=\epsilon$ is given by
\begin{align}\label{redundancy2}
    \mathcal{\beta}^{o}_{e,k}=
    \left(\frac{{3{\left(\frac{\sqrt{\xi_k}\lambda_e}
    {2(1-\epsilon)\psi_k}
    +\sqrt{\frac{{\xi_k}\lambda_e^2}
    {4(1-\epsilon)^2\psi_k^2}+\frac{1}{27}}
    \right)^{{2}/{3}}}}-1}
    {3\sqrt{\xi_k}\left(\frac{\sqrt{\xi_k}\lambda_e}
    {2(1-\epsilon)\psi_k}
    +\sqrt{\frac{{\xi_k}\lambda_e^2}
    {4(1-\epsilon)^2\psi_k^2}+\frac{1}{27}}
    \right)^{{1}/{3}} }\right)^2.
\end{align}}
\end{proposition}

\begin{proposition}\label{beta_e_proposition2}
\textit{In the large $\epsilon$ regime, as $M_k\rightarrow\infty$, the root of $\mathcal{P}^{int}_{s,k}(\beta_{e,k})=\epsilon$ is given by
\begin{equation}\label{redundancy2}
  \mathcal{\beta}^{\star}_{e,k}={\delta}
  \frac{\phi_k}{1-\phi_k}
  \ln  \left(\frac{\frac{\alpha}{2}
  \frac{1-\phi_k}{\phi_k}\left(\frac{ \psi_k(1-\epsilon)}{\lambda_e}\right)^{-\frac{\alpha}{2}}}
  {\mathcal{W}\left(\frac{\alpha}{2}
  \frac{1-\phi_k}{\phi_k}\left(\frac{ \psi_k(1-\epsilon)}{\lambda_e}\right)^{-\frac{\alpha}{2}}\right)}\right),
\end{equation}
where $\mathcal{W}(x)$ is the
Lambert-$W$ function \cite[Sec. 4-13]{Olver2010NIST}.}
\end{proposition}

To demonstrate the accuracy of $\mathcal{R}^{\star}_{e,k}\triangleq \log (1+\beta^{\star}_{e,k})$ on $\mathcal{R}^{*}_{e,k}$, we define $\Delta \mathcal{R}_{e,k}\triangleq \frac{|\mathcal{R}^{\star}_{e,k}-\mathcal{R}^*_{e,k}|}
    {\mathcal{R}^*_{e,k}}$.
   Numerically, we obtain $\Delta \mathcal{R}_{e,1}=0.0462$ when $M_1=4$, and $\Delta \mathcal{R}_{e,1}= 0.0062$ when $M_1=20$.
    This suggests that $\mathcal{R}^{\star}_{e,k}$ becomes very close to $\mathcal{R}^{*}_{e,k}$ for a large enough $M_k$ (e.g., $M_k\ge20$).

Substituting $\beta^*_{t,k}$ and $\beta^*_{e,k}$ into \eqref{st}, we obtain an expression of $\mathcal{T}$.
\begin{figure}[!t]
\centering
\includegraphics[width=3.0in]{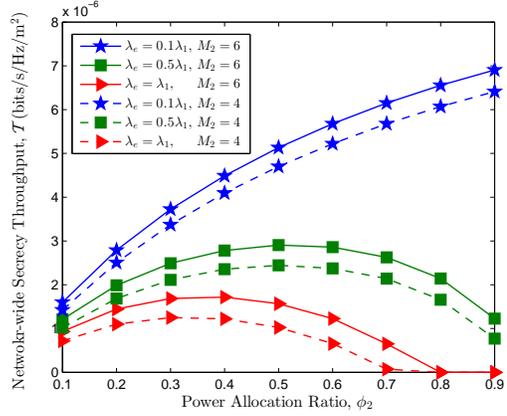}
\caption{Secrecy throughput in a 2-tier HCN vs. $\phi_2$ for different $\lambda_e$'s and $M_2$'s, $\{P_1,P_2\}=\{30,10\}$dBm, $M_1=6$, $\phi_1 = 0.8$, $\{\lambda_2,\lambda_u\} = \{2,4\}\lambda_1$, $\tau = -90$dBm, $\varrho=0.9$, and $\epsilon = 0.95$.}
\label{TH_Phi_Le}
\end{figure}
Fig. \ref{TH_Phi_Le} illustrates the network-wide secrecy throughput $\mathcal{T}$ versus $\phi_2$ for different values of $\lambda_e$ and $M_2$.
As expected, using more transmit antennas always increases $\mathcal{T}$.
We observe that, for a small $\lambda_e$, allocating more power to the information signal (increasing $\phi_2$) improves $\mathcal{T}$.
However, for a larger $\lambda_e$, $\mathcal{T}$ first increases and then decreases as $\phi_2$ increases, and even vanishes for too large a $\phi_2$ (e.g., $\lambda_e=\lambda_1$, and $\phi_2=0.8$).
There exists an optimal $\phi_2$ that maximizes $\mathcal{T}$, which can be numerically calculated by taking the maximum of $\mathcal{T}$.
We also observe that the optimal $\phi_2$ decreases as $\lambda_e$ increases, i.e., more power should be allocated to the artificial noise to achieve the maximum $\mathcal{T}$.

From the analysis in previous sections, we find that the access threshold triggers a non-trivial tradeoff between link quality and network-wide secrecy throughput.
On one hand, setting a small access threshold improves spatial reuse by enabling more communication links per unit area, potentially increasing throughput; 
setting a small access threshold also benefits secrecy since more BS are activated increasing artificial noise to impair eavesdroppers.
On the other hand, the additional amount of interference caused by the increased concurrent transmissions (a small value of $\tau$) degrades ongoing legitimate transmissions, decreasing the probability of successfully connecting the BS-UE pairs.
In this regard, neither too large nor too small an access threshold can yield a high secrecy throughput.
As shown in Fig. \ref{TH_TAU_LE_LA}, $\mathcal{T}$ first increases and then decreases as $\tau$ increases.
Only by a proper choice of the access threshold, can we achieve a high network-wide secrecy throughput.

In view of the quasi-concavity of $\mathcal{T}$ w.r.t. $\tau$ indicated in Fig. \ref{TH_TAU_LE_LA}, we can seek out the optimal $\tau$ that maximizes $\mathcal{T}$ using the gold section method\footnote{Due to the complicated expression of $\mathcal{T}$ w.r.t. $\tau$, more efficient methods for optimizing $\tau$ will be left for future research.}.
Furthermore, combined with the asymptotic analysis on $\mathcal{P}^{int}_{c,k}$ in Sec. III-C and the expression of $\mathcal{P}_{s,k}^{int}$ in
\eqref{ps_k_bound_2}, we directly provide the following asymptotic behaviors of $\mathcal{T}$ as $\tau$ goes to zero and as $\tau$ goes to infinity:
1) \textit{When all tiers share the same values of $M$ and $\phi$, and $\lambda_u\gg \lambda_j$, $\forall j\in\mathcal{K}$, $\mathcal{T}$ converges to a constant value as $\tau\rightarrow0$};
2) \textit{$\mathcal{T}\rightarrow 0$ as $\tau\rightarrow\infty$.}

\begin{figure}[!t]
\centering
\includegraphics[width=3.0in]{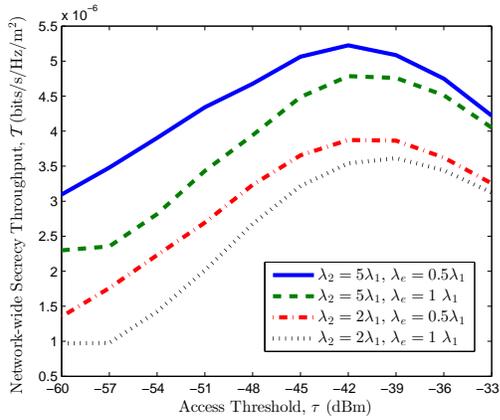}
\caption{Network-wide secrecy throughput in a 2-tier HCN vs. $\tau$ for different $\lambda_2$'s and $\lambda_e$'s, with $\{P_1,P_2\}=\{30,10\}$dBm, $\{M_1,M_2\}=\{4,4\}$, $\{\phi_1,\phi_2\}=\{0.5,0.5\}$, $\lambda_u = 10\lambda_1$, $\varrho=0.95$, and $\epsilon = 0.95$.}
\label{TH_TAU_LE_LA}
\end{figure}

\begin{figure}[!t]
\centering
\includegraphics[width=3.0in]{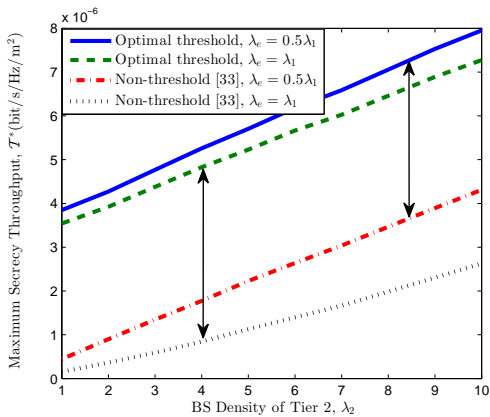}
\caption{Network-wide secrecy throughput in a 2-tier HCN vs. $\tau$ for different $\lambda_2$'s and $\lambda_e$'s, with $\{P_1,P_2\}=\{30,10\}$dBm, $\{M_1,M_2\}=\{6,4\}$, $\{\phi_1,\phi_2\}=\{0.5,0.5\}$, $\lambda_u = 10\lambda_1$, $\varrho=0.95$, and $\epsilon = 0.95$.}
\label{TS_COMPARE}
\end{figure}

Fig. \ref{TS_COMPARE} compares the network-wide secrecy throughput obtained under the optimal access threshold with that under a non-threshold policy \cite{Jo2012Heterogeneous}.
Obviously, our threshold-based policy significantly improves the secrecy throughput performance of the HCN.
We also observe that deploying more picocells still benefits network-wide secrecy throughput, even through it increases network interference.
 This is because of cell densification and the fact that the increased interference also degrades the wiretap channels.

\subsection{Average User Secrecy Throughput}
Given that each BS adopts TDMA with equal time slots allocated to the associated UEs in a round-robin manner, here we investigate the average user secrecy throughput, which is defined as
\begin{equation}\label{user_st}
    \mathcal{T}_{u,k}\triangleq \frac{\mathcal{T}_k}{\mathcal{N}_k},
\end{equation}
where $\mathcal{T}_k\triangleq
\mathcal{A}_k\varrho\left[\mathcal{R}^*_{t,k} - \mathcal{R}^*_{e,k}\right]^{+}$ denotes the secrecy transmission capacity of a cell in tier $k$, and $\mathcal{N}_k = \frac{\lambda_u}{\lambda_k}\mathcal{S}_k$ denotes the corresponding cell load.
From the perspective of a UE, the network-wide secrecy throughput can be alternatively expressed as
\begin{equation}\label{T_u}
  \mathcal{T}_u = \sum_{k\in\mathcal{K}} \lambda_u\mathcal{T}_{u,k}\mathcal{S}_k.
\end{equation}
By substituting \eqref{user_st} into \eqref{T_u}, we see that $\mathcal{T}_u=\mathcal{T}$, which is also as expected.

A UE may be interested in the minimum level of secrecy throughput it can achieve.
To this end, we evaluate the minimum average secrecy throughput over all tiers, which is defined as
\begin{equation}\label{min_st}
  \mathcal{T}_{min} \triangleq \min_{k\in\mathcal{K}} \mathcal{T}_{u,k}.
\end{equation}
Fig. \ref{MTH_M_Tau_Lu} shows how $\mathcal{T}_{min}$ depends on $M_1$, $\lambda_u$, and $\lambda_e$, respectively.
Obviously, average user secrecy throughput deteriorates as the density of Eves increases.
 This is ameliorated by adding more transmit antennas at BSs.
In addition, as $\lambda_u$ increases, the number of UEs sharing limited resources increases, which results in a decrease in per user secrecy throughput.

\begin{figure}[!t]
\centering
\includegraphics[width=3.0in]{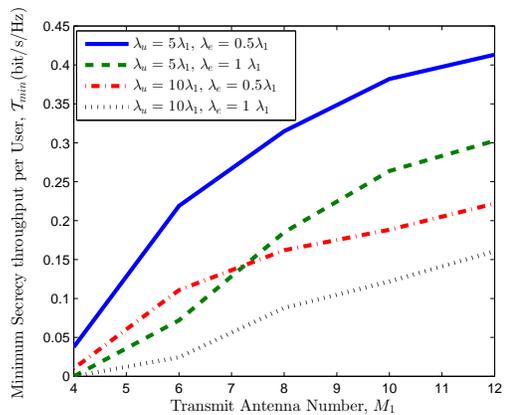}
\caption{Minimum secrecy throughput per user in a 2-tier HCN vs. $M_1$ for different $\lambda_u$'s and $\lambda_e$'s, with $\alpha=4$, $\{P_1,P_2\}=\{30,10\}$dBm, $M_2=4$, $\lambda_1=\frac{1}{\pi 400^2\mathrm{m}^2  }$,  $\lambda_2=2\lambda_1$, $\{\phi_1,\phi_2\}=\{0.5,0.5\}$, $\tau=-60$dBm, $\varrho=0.95$, and $\epsilon = 0.95$.}
\label{MTH_M_Tau_Lu}
\end{figure}

\section{Conclusions}
    This paper comprehensively studies the PLS of HCNs where the locations of all BSs, UEs and Eves are modeled as independent homogeneous PPPs.
    We first propose a mobile association policy based on the truncated ARSP and derive the tier association and BS activation probabilities.
    We then analyze the connection and secrecy probabilities of the artificial-noise-aided secure transmission.
    For connection probability, we provide a new accurate integral formula as well as an analytically tractable expression for the interference-limited case.
    For secrecy probability, we obtain closed-form expressions for both upper and lower bounds, which are approximate to the exact values at the high secrecy probability regime.
    Constrained by the connection and secrecy probabilities, we evaluate the network-wide secrecy throughput and also the minimum secrecy throughput per user.
    Numerical results are presented to verify our theoretical analysis.

\appendix
\subsection{Proof of Lemma \ref{association_probability_lemma}}
\label{association_probability_proof}
    Consider a typical UE $\textsf{U}_o$, we know from \eqref{arsp} and \eqref{sk_def} that $\textsf{U}_o$ is associated with tier $k$ only when $r_k\le D_k$ and $ \hat{P}_k>\hat{P}_j, \forall j\in\mathcal{K}\setminus k$ simultaneously hold.
    Therefore, $\mathcal{S}_k$ can be calculated as
    \begin{equation}\label{ass_pro_appendix}
        \mathcal{S}_k
        =\int_0^{D_k}\prod_{j\in\mathcal{K}\setminus k} \mathbb{P}\left\{\hat{P}_k>\hat{P}_j|r_k \right\}f_{r_k}(r)dr,
    \end{equation}
    where $f_{r_k}(r)=2\pi\lambda_k re^{-\pi\lambda_kr^2}$ \cite[Theorem 1]{Haenggi08Distance}.
    The term $\mathbb{P}\left\{\hat{P}_k>\hat{P}_j|r_k \right\}$ can be calculated as
   \begin{align}\label{Pr_kj}
    \mathbb{P}&\left\{\hat{P}_k>\hat{P}_j|r_k \right\}
\stackrel{\mathrm{(c)}}=
        \mathbb{P}\left\{r_j>\min\left\{D_j,\left({P}_{j,k}
        {M}_{j,k}\right)^{\frac{1}{\alpha}}r_k\right\}|r_k\right\}\nonumber\\
        &\stackrel{\mathrm{(d)}}=
        \mathbb{P}\left\{r_j>\left({P}_{j,k}
        {M}_{j,k}\right)^{\frac{1}{\alpha}}r_k\right\}\nonumber\\
        &=\mathbb{P}\left\{\textit{No BS in tier j is inside } \mathcal{B}\left(o, \left({P}_{j,k}
        {M}_{j,k}\right)^{\frac{1}{\alpha}}r_k\right)\right\}\nonumber\\
        &\stackrel{\mathrm{(e)}} = e^{-\pi\lambda_j \left({P}_{j,k}{M}_{j,k}\right)^{{\delta}}r_k^2},
   \end{align}
        where (c) follows from the fact that $\hat{P}_k>\hat{P}_j$ holds if $r_j>D_j$ or $\frac{P_jM_j}{r_j^{\alpha}}<\frac{P_kM_k}{r_k^{\alpha}}
        \Rightarrow r_j>\left({P}_{j,k}
        {M}_{j,k}\right)^{\frac{1}{\alpha}}r_k$,
        (d) holds for $D_j = \left(\frac{P_jM_j}{\tau}\right)
^{\frac{1}{\alpha}} = \left({P}_{j,k}{M}_{j,k}\right)
^{\frac{1}{\alpha}}D_k\ge \left({P}_{j,k}{M}_{j,k}\right)
^{\frac{1}{\alpha}}r_k$, and (e) is obtained from the basic nature of PPP \cite{Andrews2011Tractable}.
    Substituting \eqref{Pr_kj} into \eqref{ass_pro_appendix} and calculating the integral, we complete the proof.

\subsection{Proof of Lemma \ref{activation_probability_lemma}}
\label{activation_probability_proof}
Without lose of generality, we consider a tagged BS $\textsf{B}_{z,k}$.
Let $\Phi_u^{o}\triangleq\Phi_u\bigcap \mathcal{B}(z,D_k)$, and the BS idle probability of tier $k$, i.e., the complementary probability $\bar{\mathcal{A}}_k\triangleq 1-\mathcal{A}_k$, can be calculated as follows
    \begin{align}\label{A_k_inverse}
        \bar{\mathcal{A}}_k
        &=\mathbb{E}_{\Phi_u}\left[\prod_{x\in
        \Phi_u^{o}} \mathbb{P}\{\textsf{U}_x \textit{ is not associated with }\textsf{B}_{z,k}\}\right]\nonumber\\
       & =\mathbb{E}_{\Phi_u}\left[\prod_{x\in\Phi_u^{o}} \mathbb{P}\left\{\frac{P_kM_k}{r_{zx}^{\alpha}}<\max_{j\in\mathcal{K}
        }\hat{P}_j \right\}\right]\nonumber\\
       & \stackrel{\mathrm{(f)}}
       =\mathbb{E}_{\Phi_u}\left[\prod_{x\in\Phi_u^{o}}
        1-e^{-\pi\Xi\left(P_kM_k\right)^{-{\delta}}
    r_{zx}^2}\right]\nonumber\\
    &\stackrel{\mathrm{(g)}}=\exp\left(-2\pi\lambda_u\int_0^{D_k}
  e^{-\pi\Xi\left(P_kM_k\right)^{-{\delta}}r^2}rdr\right),
    \end{align}
    where (f) is obtained from \eqref{Pr_kj} and (g) is derived by using the probability generating functional lemma (PGFL) over PPP \cite{Stoyan1996Stochastic}.
    Solving the integral term in \eqref{A_k_inverse} yields $\mathcal{A}_k =1-\bar{\mathcal{A}}_k$, which is given in \eqref{A_k}.

\subsection{Proof of Property \ref{BS_density_property}}
\label{BS_density_proof}
Since $\mathcal{A}_k = 1-e^{-\frac{\lambda_u}{\lambda_k}\mathcal{S}_k}$, the first derivative of $\lambda_k^{o}=\mathcal{A}_k\lambda_k$ on $\lambda_k$ can be given by
    \begin{equation}\label{d_lk}
        \frac{d\lambda_k^{o}}{d\lambda_k}
        =1-e^{-\frac{\lambda_u}{\lambda_k}\mathcal{S}_k}
        \left(1+\frac{\lambda_u}{\lambda_k}\mathcal{S}_k
        -\lambda_u \frac{d\mathcal{S}_k}{d\lambda_k}\right).
    \end{equation}
    Let $\Psi_k\triangleq \frac{(P_kM_k)^{{\delta}}}{\Xi}
    \left(1-e^{-\pi\tau^{-{\delta}}\Xi}\right)$ and thus $\mathcal{S}_k=\lambda_k\Psi_k$.
We easily see from \eqref{S_k} that $\mathcal{S}_k$ monotonically increases on $\lambda_k$, which means $\frac{d\mathcal{S}_k}{d\lambda_k}=\Psi_k
    +\lambda_k\frac{d\Psi_k}{d\lambda_k}>0\Rightarrow
    -\frac{d\Psi_k}{d\lambda_k}<\frac{\Psi_k}{\lambda_k}$.
    Substituting this inequality into \eqref{d_lk} yields $\frac{d\lambda_k^{o}}{d\lambda_k}>1-
    e^{-\frac{\lambda_u}{\lambda_k}\mathcal{S}_k}
    \left(1+\frac{\lambda_u}{\lambda_k}\mathcal{S}_k\right)$. By employing $1+x<e^x$, we obtain $\frac{d\lambda_k^{o}}{d\lambda_k}>0$, i.e., $\lambda_k^{o}$ monotonically increases on $\lambda_k$.
    As $\lambda_k\rightarrow\infty$, we have $\mathcal{S}_k\rightarrow 1$, and then $\lambda_k^{o} \rightarrow\lim_{\lambda_k\rightarrow \infty}\lambda_k
    \left(1-e^{-\frac{\lambda_u}{\lambda_k}}\right)=\lambda_u$.

\subsection{Proof of Theorem \ref{pck_theorem}}
\label{pck_theorem_proof}
Define $x_p \triangleq \frac{(-1)^p s^p}
    {p!}
    \mathcal{L}^{(p)}_{I_o}(s)$, and $\mathbf{x}_{M_k-1}\triangleq[x_1,x_2,\cdots,x_{M_k-1}]^{\mathrm{T}}$.
    $\mathcal{P}_{c,k}$ in \eqref{pc_k1} can be rewritten as
\begin{equation}\label{pc_app}
  \mathcal{P}_{c,k}=\sum_{m=0}^{M_k-1}\mathbb{E}_{R_k}\left[
  e^{-sN_0}\sum_{p=0}^{m}\frac{(sN_0)^{m-p}}{(m-p)!}x_p\right].
\end{equation}

Due to the independence of $I_{io}$ and $I_{jo}$ for $i\neq j$, the Laplace transform of $I_o$ is given by
\begin{align}\label{La_Io_app}
    \mathcal{L}_{I_o}(s) = \mathbb{E}_{I_o}\left[e^{-sI_o}\right]
    =\prod_{j\in\mathcal{K}} \mathcal{L}_{I_{jo}}(s).
\end{align}

Let $r_{jo}\triangleq
\left({P_{j,k}M_{j,k}}\right)^{\frac{1}{\alpha}}R_k$ and $P_{jz}\triangleq\phi_jP_j \left(|\mathbf{h}_{zo}^{\mathrm{T}}\mathbf{w}_{z}|^2
+
\xi_j\|\mathbf{h}_{zo}^{\mathrm{T}}
\mathbf{W}_{z}\|^2\right)$.
Using \cite[(8)]{Haenggi2009Stochastic}, $\mathcal{L}_{I_o}(s)$ can be calculated as
\begin{align}\label{La_jo_app}
\mathcal{L}_{I_{o}}(s)&= \prod_{j\in\mathcal{K}} \mathbb{E}_{\Phi_j}
    \left[\exp\left(-s\sum_{z\in\Phi^o_j\setminus \mathcal{B}(o, r_{jo})}
    P_{jz}r_{zo}^{-\alpha}
    \right)\right]\nonumber\\
    &=\exp\left(-\pi\sum_{j\in\mathcal{K}}\lambda_j^o\int^{\infty}_{r_{jo}^2}
    \left(1-\varpi\left(P_{jz}\right)\right)dr\right),
\end{align}
where $\varpi(P_{jz})=\int_0^{\infty}
    e^{-sxr^{-\alpha}}f_{P_{jz}}(x)dx $ can be obtained by invoking $f_{P_{jz}}(x)$ in
    \cite[Lemma 1]{Zhang2013Enhancing}
\begin{align}\label{Epz}
  \varpi(P_{jz})=\begin{cases}
    ~{\left(1+\omega
    r^{-\frac{\alpha}{2}}\right)^{-M_j}}, & \xi_j=1,\\
    ~\frac{\left(1-\xi_j\right)^{1-M_j}}
    {1+\omega r^{-\frac{\alpha}{2}}}
    -\sum\limits_{n=0}^{M_j-2}
    \frac{\xi_j
    \left(1-\xi_j\right)^{1-M_j+n}}
    {\left(1+\xi_j\omega r^{-\frac{\alpha}{2}}
    \right)^{n+1}}, &\xi_j\neq 1.
    \end{cases}
\end{align}
with $\omega\triangleq \phi_jP_j s$.
Next, we calculate $\mathcal{L}^{(p)}_{I_o}(s)$.
We consider the case $\phi_j\neq\frac{1}{M_j}$, and the results for $\phi_j=\frac{1}{M_j}$ can be easily obtained in a similar way.
We present $\mathcal{L}^{(p)}_{I_o}(s)$ in the following recursive form
\begin{align}\label{p_derivative}
    &\mathcal{L}^{(p)}_{I_o}(s)=\pi\sum_{j\in\mathcal{K}}\lambda_j^o
\sum_{i=0}^{p-1}\binom{p-1}{i}
    \frac{(-1)^{p-i}}{(1-\xi_j)^{M_j-1}}\mathcal{L}^{(i)}_{I_o}(s)\times\nonumber\\
    &\qquad\Bigg\{
    \int^{\infty}_{r_{jo}^2}\Bigg(\frac{(p-i)!\left(
    \phi_jP_jr^{-\frac{\alpha}{2}}\right)^{p-i}}{\left(
    1+\omega r^{-\frac{\alpha}{2}}\right)^{    p-i+1}}-\nonumber\\
    &\quad\sum_{n=0}^{M_j-2}\frac{\xi_j(n+p-i)!\left(
    \xi_j\phi_jP_jr^{-\frac{\alpha}{2}}\right)^{p-i}}
    {n!(1-\xi_j)^{-n}\left(
    1+\xi_j\omega r^{-\frac{\alpha}{2}}\right)^{n+p-i+1}}
    \Bigg)dr\Bigg\}.
\end{align}
Using the variable transformation $r^{-\frac{\alpha}{2}}\rightarrow v$, and plugging $\mathcal{L}^{(p)}_{I_o}(s)$ into $x_p$, we have for $p\ge1$
\begin{align}\label{xp}
x_p &=\sum_{i=0}^{p-1}\Bigg\{
    \frac{p-i}{p}\sum_{j\in\mathcal{K}}
{\frac{\pi{\delta}\lambda_j^o\omega^{p-i}}{(1-\xi_j)^{M_j-1}}
\int_0^{r_{jo}^{-\alpha}}\Bigg(\frac{v^{p-i-{\delta}-1}}
    {(1+\omega v)^{p-i+1}}}\nonumber\\
    &-{\sum_{n=0}^{M_j-2}\binom{n+p-i}{n}
    \frac{\xi_j^{p-i+1}(1-\xi_j)^{n}v^{p-i-{\delta}-1}}
    {(1+\xi_j\omega v)^{n+p-i+1}}\Bigg)dv}\Bigg\}x_i.
\end{align}
Calculating \eqref{xp} with \cite[(3.194.1)]{Gradshteyn2007Table}, and after some algebraic manipulations, $x_p$ can be given as
\begin{equation}\label{xp_2}
  x_p =R_k^2\sum_{i=0}^{p-1}\frac{p-i}{p}
    g_{p-i}x_i,
\end{equation}
where $g_{i}\triangleq\frac{\pi{\delta}}{i-{\delta}}\sum_{j\in\mathcal{K}}\lambda_j^o
\left({ P_{j,k}M_{j,k}}\right)^{{\delta}}
\left(\frac{\phi_jP_j}{r_{jo}^{\alpha}}s\right)^i
\mathcal{Q}_{j,i}$, with $\mathcal{Q}_{j,i}$ given in \eqref{Q}.
$x_0$ in \eqref{xp_2} is calculated as
\begin{align}\label{x0}
    &x_0=\mathcal{L}_{I_o}(s)=
    \exp\Bigg(-\pi\sum_{j\in\mathcal{K}}\lambda_j^o\times\nonumber\\
    &\Bigg(  \underbrace{\int^{\infty}_{0}
    \left(1-\varpi(P_{jz})\right)dr}_{\mathcal{I}_{j1}(s)}-
    \underbrace{\int^{r_{jo}^2}_{0}
    \left(1-\varpi(P_{jz})\right)dr}_{\mathcal{I}_{j2}(s)}
    \Bigg)\Bigg).
\end{align}
$\mathcal{I}_{j1}(s)$ can be directly obtained from \cite[(8)]{Haenggi2009Stochastic}, i.e.,
\begin{equation}\label{I_j1_app}
    \mathcal{I}_{j1}(s)=
    \begin{cases}
    ~\omega^{{\delta}}C_{\alpha,M_j+1}, & \xi_j=1,\\
    \frac{\omega^{{\delta}}C_{\alpha,2}}{
    \left(1-\xi_j\right)^{M_j-1}}-
    \sum\limits_{n=0}^{M_j-2}\frac{\omega^{{\delta}}
    \xi_j^{1+{\delta}}C_{\alpha,n+2}}
    {\left(1-\xi_j\right)^{M_j-1-n}}, & \xi_j\neq1.
    \end{cases}
\end{equation}
$\mathcal{I}_{j2}(s)$ can be derived by invoking $\varpi(P_{jz})$ in \eqref{Epz}.
Specifically, when $\xi_j=1$,
    \begin{align}\label{I2_1_app}
    \mathcal{I}_{j2}(s)=
    r_{jo}^{2}
    \left[1-{\delta}\frac{
    {_2}F_1\left(M_j,M_j+\delta;M_j+\delta+1;
    -\frac{r_{jo}^{\alpha}}{\omega}\right)}
    {\left(M_j+{\delta}\right)(\omega r_{jo}^{-\alpha})^{M_j}}
\right],
    \end{align}
and when $\xi_j\neq1$,
\begin{align}\label{I2_2_app}
    &\mathcal{I}_{j2}(s)=r_{jo}^{2}
\Bigg[1-{\delta}
    \Bigg(\frac{{_2}F_1\left(1,{\delta}+1;{\delta}+2;
    -(\omega r_{jo}^{-\alpha})^{-1}\right)}
    {\left(1+{\delta}\right){(1-\xi_j)^{M_j-1}}
    (\omega r_{jo}^{-\alpha})^{-1}}-\nonumber\\
    &\sum_{n=0}^{M_j-2}\frac{{_2}F_1\left(
    n+1,n+1+{\delta};n+2+{\delta};
    -(\xi_j\omega r_{jo}^{-\alpha})^{-1}\right)}
    {\left(n+1+{\delta}\right){(1-\xi_j)^{M_j-1-n}}
    \left(\xi_j\omega r_{jo}^{-\alpha}
    \right)^{n}}\Bigg)\Bigg].
\end{align}

Having obtained a linear recurrence form for $x_p$ in \eqref{xp_2}, similar to \cite{Li2014Throughput},
$\mathbf{x}_{M_k}$ can be given by
\begin{equation}\label{xp_matrix}
  \mathbf{x}_{M_k-1}=\sum_{i=1}^{M_k-1}R_k^{2i}x_0
  \mathbf{G}_{M_k-1}^{i-1}\mathbf{g}_{M_k-1},
\end{equation}
where $\mathbf{g}_{M_k}$ and
$\mathbf{G}_{M_k}$ have the same forms as those in \cite{Li2014Throughput}.
From \cite[(39)]{Li2014Throughput}, we have $\mathbf{G}_{M_k-1}^{i-1}\mathbf{g}_{M_k-1}=\frac{1}{i!}\mathbf{\Theta}_{M_k}^i(2:M_k,1)$, where $\mathbf{\Theta}_{M_k}^i(2:M_k,1)$ represents the entries from the second to the $M_k$-th row in the first column of $\mathbf{\Theta}_{M_k}^i$, with $\mathbf{\Theta}_{M_k}$ shown in \eqref{QM}.
Then $x_p$ can be expressed as
\begin{equation}\label{xp_3}
  x_p = \sum_{i=0}^{M_k-1}R_k^{2i}x_0\frac{1}{i!}\mathbf{\Theta}_{M_k}^i(p+1,1),
\end{equation}
and consequently, $\mathcal{P}_{c,k}$ can be given by
\begin{equation}\label{pc_app_final}
  \mathcal{P}_{c,k} = \sum_{m=0}^{M_k-1}\sum_{p=0}^{m}\sum_{i=0}^{M_k-1}
  \mathbb{E}_{R_k}\left[
  x_0\frac{e^{-sN_0}R_k^{2i}\mathbf{\Theta}_{M_k}^i(p+1,1)}
  {(m-p)!i!(sN_0)^{p-m}}\right].
\end{equation}
To calculate the above expectation, we give the PDF of $R_k$ in the following lemma
\begin{lemma}\label{pdf_rk}
\textit{The PDF of $R_k$ is given by
\begin{equation}\label{F_Dk}
  f_{R_k}(x) = \begin{cases}
  ~ \frac{2\pi\lambda_k}{\mathcal{S}_k}x\exp\left(
  -\pi\Xi\left(P_kM_k\right)^{-{\delta}}x^2\right), & x\le D_k,\\
  ~0,& x>D_k.
  \end{cases}
\end{equation}}
\end{lemma}
\begin{proof}
    The result for $x\le D_k$ is obtained from  \cite[Lemma 3]{Jo2012Heterogeneous}, while that for $x> D_k$ is an immediate consequence of \eqref{arsp}.
\end{proof}

Averaging over $R_k$ using \eqref{F_Dk} completes the proof.

\subsection{Proof of Corollary \ref{pck_corollary}}
\label{pck_corollary_proof}

For the interference-limited case where $N_0=0$, \eqref{pc_app_final} can be simplified as
\begin{equation}\label{pc_app_inter_final}
  \mathcal{P}^{inter}_{c,k} = \sum_{m=0}^{M_k-1}\sum_{i=0}^{M_k-1}
  \mathbb{E}_{R_k}\left[
 \frac{1}{i!}
  { x_0 R_k^{2i}}
  \mathbf{\Theta}_{M_k}^i(m+1,1)\right],
\end{equation}
which can be alternatively expressed in the following form using the $L_1$ induced matrix norm,
\begin{equation}\label{pc_app_inter_final2}
  \mathcal{P}^{inter}_{c,k} = \sum_{i=0}^{M_k-1}
  \mathbb{E}_{R_k}\left[
  \left\|\frac{1}{i!}
  {x_0 R_k^{2i}}
  \mathbf{\Theta}_{M_k}^i\right\|_1\right].
\end{equation}
Averaging over $R_k$ using \eqref{F_Dk} completes the proof.

\subsection{Proof of Properties \ref{pc_property2}-\ref{pc_property5}}
\label{proof_property_2_5}
\textit{i. Proof of Property \ref{pc_property2}:}
    For the case $\{M_j\}=M$, $\{\phi_j\}=\phi$ and $\lambda_u\gg \lambda_j$, $\forall j\in\mathcal{K}$, $\mathcal{A}_j\rightarrow1$ as $\tau\rightarrow0$;
    $\Upsilon_k$ and $\mathbf{\Theta}_{M_k}$ can be re-expressed as  $\Upsilon_k=\frac{\Xi}{(P_kM)^{\delta}}\tilde \Upsilon_0$, and $\mathbf{\Theta}_{M_k}=\frac{\Xi}{(P_kM)^{\delta}}\tilde {\mathbf{\Theta}}_{M}$, where both $\tilde \Upsilon_0$ and $\tilde {\mathbf{\Theta}}_{M}$ are independent of $P_j$, $\lambda_j$ and $k$.
    Since $D_k\rightarrow\infty$ as $\tau\rightarrow0$, by omitting the
term $\sum_{l=0}^m\frac{\pi^{l}e^{-\pi\Upsilon_k D_k^2}}
     {l!D_k^{-2l}\Upsilon^{-l}_k}$ from \eqref{pc_inter} and substituting in $\Upsilon_k$, $\mathbf{\Theta}_{M_k}$ along with $\mathcal{S}_k^{\epsilon=0}
    =\frac{\lambda_k(P_kM)^{\delta}}{\Xi}$, we obtain
\begin{align}\label{pc_epsilon_0}
     \mathcal{P}^{int,\epsilon=0}_{c,k}&=
     \sum_{m=0}^{M-1}
     \frac{1}
     {\pi^{m}\tilde\Upsilon^{m+1}_0}
     \left\|\tilde{\mathbf{\Theta}}_{M}^m\right\|_1,
     ~~\forall k\in\mathcal{K}
\end{align}
which is obviously independent of $P_j$, $\lambda_j$ and $k$.

\textit{ii. Proof of Property \ref{pc_property3}:}
To complete the proof, we first give the following lemma.
\begin{lemma}\label{pc_bound_lemma}
\textit{For the interference-limited HCN, the UCP of a typical UE associated with tier $k$ satisfies
\begin{equation}\label{pc_bound_def}
  \mathcal{P}^{int,B}_{c,k}(\beta_t)\le\mathcal{P}^{int}_{c,k}
  \le\mathcal{P}^{int,B}_{c,k}(\varphi_{k}\beta_t),
\end{equation}
where $\varphi_{k}\triangleq (M_k!)^{-\frac{1}{M_k}}$ and
\begin{align}\label{pc_bound_inter}
  \mathcal{P}^{int,B}_{c,k}(\beta) =
  \frac{\lambda_k}{\mathcal{S}_k}\sum_{m=1}^{M_k}
  \binom{M_k}{m}\frac{(-1)^{m+1}}
  {\hat \Upsilon_{k,m\beta}}
  \left(1-e^{-\pi\hat\Upsilon_{k,m\beta}D_k^2}\right).
\end{align}
with the value of $\hat \Upsilon_{k,m\beta}$ equal to that of $\Upsilon_{k}$ at $\beta_t=m\beta$.
}
\end{lemma}
\begin{proof}
    Recalling \eqref{pc_k1}, since $\|\mathbf{h}_{b}\|^2\sim \Gamma(M_k,1)$, we have  $\mathbb{P}\left\{\|\mathbf{h}_{b}\|^2\geq x\right\}=
    1-\int_0^x\frac{e^{-v}v^{M_k-1}}{(M_k-1)!}dv$, which can be rewritten as the form of $1-\frac{1}{\Gamma(1+1/t)}\int_0^z e^{-v^t}dv$ with $t=1/M_k$ and $z=x^{M_k}$.
    Then according to Alzer's inequality \cite{Alzer1997Mathematics}, we obtain the following relationship
    \begin{equation}\label{ccdf_bound}
 1- \left(1-e^{-x}\right)^{M_k}\le \mathbb{P}\left\{\|\mathbf{h}_{b}\|^2\geq x\right\}
 \le 1- \left(1-e^{-\varphi_kx}\right)^{M_k}.
\end{equation}
Substituting \eqref{ccdf_bound} into $\mathcal{P}^{int}_{c,k}= \mathbb{E}_{R_k}\mathbb{E}_{I_o} \left[ \mathbb{P}\left\{\|\mathbf{h}_{b}\|^2\geq {s}I_o\right\}\right]$, we can finally obtain \eqref{pc_bound_def} and \eqref{pc_bound_inter}.
\end{proof}
As $\tau\rightarrow\infty$, we have $\mathcal{A}_k\rightarrow 0$ for $k\in\mathcal{K}$.
 Accordingly, we obtain $\Upsilon_k=\frac{\Xi}{(P_{k}M_k)^{\delta}}$, which becomes independent of $\beta_t$, and so does $\hat \Upsilon_{k,m\beta}$.
This implies $\mathcal{P}^{int,B}_{c,k}(\beta_t)
=\mathcal{P}^{int,B}_{c,k}(\varphi_{k}\beta_t)
=\mathcal{P}^{int}_{c,k}$.
Substituting \eqref{S_k} into \eqref{pc_bound_inter}, $\mathcal{P}^{int}_{c,k}$ can be finally reduced to $\sum_{m=1}^{M_k}\binom{M_k}{m}(-1)^{m+1}=1$, which completes the proof.

\textit{iii. Proof of Property \ref{pc_property4}:}
Considering ${P_{j,1}}\rightarrow 0$, $\forall j\neq1$, we have $\mathcal{S}_j\rightarrow 0$ and $\mathcal{A}_j\rightarrow 0$, and accordingly $\Upsilon_1\propto \lambda_1\mathcal{A}_1$, $\Upsilon_j\propto \lambda_1\mathcal{A}_1P_{1,j}^{\delta}$, and $\left\|\mathbf{\Theta}_{M_1}^m\right\|_1
\propto(\lambda_1\mathcal{A}_1)^m$, $\left\|\mathbf{\Theta}_{M_j}^m\right\|_1
\propto(\lambda_1\mathcal{A}_1P_{1,j}^{\delta})^m$.
We see that, both $D_1$ and $\Upsilon_j$ goes to infinite as ${P_{j,1}}\rightarrow 0$, then by omitting the exponential term from \eqref{pc_inter}, and combined with the above observations, we obtain $\mathcal{P}^{int}_{c,1}\propto \eta_1\triangleq\frac{1}{\mathcal{S}_1 \mathcal{A}_1}$ and $\mathcal{P}^{int}_{c,j}\propto\eta_j
    \triangleq\frac{\lambda_{j,1}{ P^{{2}/{\alpha}}_{j,1}}}{\mathcal{S}_j \mathcal{A}_1}$.
    From \eqref{S_k} and \eqref{A_k}, we can readily see that both $\frac{1}{\mathcal{S}_1\mathcal{A}_1}$ and $\frac{\lambda_{j,1}}{\mathcal{S}_j \mathcal{A}_1}$ monotonically increase on $\tau$ and  $\lambda_l$, $\forall l\neq1$, while decrease on $\lambda_u$, which completes the proof.

\textit{iv. Proof of Property \ref{pc_property5}:}
    Similar to the proof for Property \ref{pc_property3},
    we have $\mathcal{P}^{int}_{c,1}\propto\eta_1$ and $\mathcal{P}^{int}_{c,j}\propto\eta_j$ as $ P_{j,1}\ll 1$.
    Since $P_1$ increases $\mathcal{S}_1$, $\mathcal{A}_1$ and $P_1\mathcal{S}_j$, we easily see that both $\eta_1$ and $\eta_j$ decrease on $P_1$.
    As $P_1\rightarrow \infty$, we obtain $\mathcal{S}_1\rightarrow 1$, $\mathcal{A}_1\rightarrow 1-e^{-\frac{\lambda_u}{\lambda_1}}$, and $\eta_j\rightarrow \frac{\lambda_1}{\mathcal{A}_1}{ M^{-{2}/{\alpha}}_{j,1}}$, and it is clear that $\mathcal{P}^{int}_{c,k}$, $\forall k\in\mathcal{K}$, becomes independent of $P_1$, which completes the proof.

\subsection{Proof of Theorem \ref{psk_theorem}}
\label{psk_theorem_proof}
Applying the PGFL over PPP along with the Jensen's inequality yields
\begin{align}\label{ps_k_app1}
    \mathcal{P}_{s,k} &=
    \mathbb{E}_{\Phi_1}\cdots
    \mathbb{E}_{\Phi_K}\Big[\exp\Big(-2\pi\lambda_e\times\nonumber\\
    &\int_0^{\infty}\mathbb{P}\left\{\mathbf{SINR}_{e,k}
    \ge\beta_{e}|\Phi_1,\cdots,\Phi_K\right\}rdr\Big)\Big]\nonumber\\
    &\ge \exp\left(-2\pi\lambda_e
    \int_0^{\infty}\mathbb{P}\left\{\mathbf{SINR}_{e,k}
    \ge\beta_{e}\right\}rdr\right).
\end{align}

Let $I_e = I_{be}+\sum_{j\in\mathcal{K}}I_{je}$ and $\kappa=\frac{ r_{be}^{\alpha}\beta_{e}}{\phi_kP_k}$, and we can calculate $\mathbb{P}\{\mathbf{SINR}_{e,k}>\beta_{e}\}$ as follows
\begin{align}\label{ps_k_app2}
   \mathbb{P}&\{\mathbf{SINR}_{e,k}>\beta_{e}\}=
   \mathbb{P}\left\{\left|\mathbf{h}_{be}^{\mathrm{T}}
   \mathbf{w}_{b}\right|^2> \kappa(I_e+N_0)\right\}\nonumber\\
   &\qquad\stackrel{\mathrm{(h)}}=\mathbb{E}_{I_e}
   \left[e^{-\kappa(I_e+N_0)}\right]
    = e^{-\kappa N_0}\mathcal{L}_{I_e}(\kappa),
\end{align}
where (h) holds because $U\triangleq \left|\mathbf{h}_{be}^{\mathrm{T}}
   \mathbf{w}_{b}\right|^2\sim\mathrm{Exp}(1)$ is independent of $I_e$.
   Note that, $U$ is also independent of $V\triangleq \|\mathbf{h}_{be}^{\mathrm{T}}
\mathbf{W}_{b}\|^2\sim\Gamma(M_k-1,1)$ due to the orthogonality of $\mathbf{w}_{b}$ and $\mathbf{W}_{b}$.

Similar to \eqref{La_Io_app}, the Laplace transform of $I_e$ can be expressed as
\begin{equation}\label{La_Ie_app}
    \mathcal{L}_{I_e}(\kappa) = \mathcal{L}_{I_{be}}(\kappa)
    \prod_{j\in\mathcal{K}}
    \mathcal{L}_{I_{je}}(\kappa).
\end{equation}
We first calculate $\mathcal{L}_{I_{be}}(\kappa)$ as
\begin{align}\label{La_Ibe}
    \mathcal{L}_{I_{be}}(\kappa) &= \mathbb{E}_{I_{be}}
    \left[e^{-\kappa I_{be}}\right]= \int_0^{\infty}e^{-\xi_k\phi_kP_kr_{be}^{-\alpha}\kappa v}f_V(v)dv\nonumber\\
    &= \left(
    1+\xi_k\phi_kP_kr_{be}^{-\alpha}\kappa \right)^{1-M_k},
\end{align}
where the last equality is obtained by invoking $f_V(v)=\frac{v^{M_k-2}e^{-v}}{\Gamma(M_k-1)}$ and using \cite[(3.326.2)]{Gradshteyn2007Table}.
We then obtain $\mathcal{L}_{I_{je}}(\kappa )$ from \cite[(8)]{Haenggi2009Stochastic}, which is given by
\begin{equation}\label{La_Ije}
  \mathcal{L}_{I_{je}}(\kappa ) = \exp\left(-\pi\lambda_j^oC_{\alpha,M_j}(\xi_j\phi_jP_j\kappa )^
  {{\delta}}\right).
\end{equation}

Substituting \eqref{La_Ibe} and \eqref{La_Ije} into \eqref{ps_k_app2} yields
\begin{align}\label{ps_k_app3}
   \mathbb{P}\{\mathbf{SINR}_{e,k}\ge\beta_{e}\}
   & = \frac{e^{-\kappa N_0}e^{-\pi\sum_{j\in\mathcal{K}}
    \mathcal{A}_j\lambda_jC_{\alpha,M_j}
    (\xi_j\phi_jP_j\kappa )^
  {{\delta}}}}{\left(1+\xi_k\beta_{e}\right)^{M_k-1}}.
\end{align}

Plugging \eqref{ps_k_app3} with $\kappa =\frac{ r_{be}^{\alpha}\beta_{e}}{\phi_kP_k}$ into \eqref{ps_k_app1}, we obtain the lower bound $\mathcal{P}_{s,k}^L$ as shown in \eqref{ps_k_lower}.

Next, we derive the upper bound $\mathcal{P}_{s,k}^U$ by only considering the nearest Eve to the serving BS.
Given a serving BS $\textsf{B}_{b,k}$ and the nearest Eve $\textsf{E}_{e}$,
$\mathcal{P}_{s,k}^U$ can be calculated by
\begin{equation}\label{ps_k_upper_app}
    \mathcal{P}_{s,k}^U = \int_0^{\infty}
\mathbb{P}\{\mathbf{SINR}_{e,k}<\beta_{e}\}f_{r_{be}}(r)dr,
\end{equation}
where $f_{r_{be}}(r)=2\pi\lambda_e re^{-\pi\lambda_er^2}$ \cite[Theorem 1]{Haenggi08Distance} and $\mathbb{P}\{\mathbf{SINR}_{e,k}<\beta_{e}\}
=1-\mathbb{P}\{\mathbf{SINR}_{e,k}\ge\beta_{e}\}$ can be directly obtained from \eqref{ps_k_app3}.
Calculating the integral yields the result as shown in \eqref{ps_k_upper}.

\subsection{Proof of Properties \ref{ps_property6}-\ref{ps_property8}}
\label{proof_property_6_8}
\textit{i. Proof of Property \ref{ps_property6}:}
     We obtain the monotonicity of $\xi_k$ and $\psi_k$ on $\lambda_e$, $\tau$, and $\phi_k$ from \eqref{ps_k_bound_2}: 1) Both $\xi_k$ and $\psi_k$ are independent of $\lambda_e$;
    2) $\psi_k$ monotonically decreases on $\tau$ and $\phi_l$, $\forall l\in\mathcal{K}\setminus k$, while $\xi_k$ is independent of $\tau$ and $\phi_l$;
    3) Both $\xi_k$ and $\psi_k$ monotonically decrease on $\phi_k$;
    4) $\psi_k$ monotonically increases on $\lambda_u$, while $\xi_k$ is independent of $\lambda_u$.
    Combining these results directly completes the proof.

\textit{ii. Proof of Property \ref{ps_property7}:}
    Considering ${M_{j,k}}\rightarrow 0$, we have $\left(1+\xi_k\beta_{e}\right)^{1-M_k}\approx e^{-\left(\phi_k^{-1}-1\right)\beta_{e}}$ and $\mathcal{A}_j\rightarrow 0$, then we obtain $1-\mathcal{P}_{s,k}^{int,o} \propto \frac{\lambda_e(M_k-1)}{\lambda_k\mathcal{A}_k C_{\alpha, M_k}}$.
    We can prove that $\frac{M_k-1}{\lambda_k\mathcal{A}_k C_{\alpha, M_k}}$ monotonically decreases on $M_k$ while increases on $\lambda_k$, which completes the proof.

\textit{iii. Proof of Property \ref{ps_property8}:}
    For an extremely large $P_k$, we have $1-\mathcal{P}^{int,o}_{s,k}\propto\chi_k\triangleq
    \frac{\lambda_e}{\lambda_k \mathcal{A}_k}$ and $1-\mathcal{P}^{int,o}_{s,j}\propto\chi_j\triangleq
    \frac{\lambda_e{P^{{2}/{\alpha}}_{j,k}}}{\lambda_j \mathcal{A}_k}$, $\forall j\neq k$.
    We can easily prove that both $\chi_k$ and $\chi_j$ decrease on $P_k$, i.e., each $\mathcal{P}^{int,o}_{s,j}$ increases on $P_k$.
    As $P_k\rightarrow \infty$, we have $\mathcal{A}_k\rightarrow 1-e^{-\frac{\lambda_u}{\lambda_k}}$, such that $\mathcal{P}^{int,o}_{s,k}$ tends to be constant.
    Besides, we have $\lim_{P_k\rightarrow\infty}\chi_j=0$, which yields $\lim_{P_k\rightarrow\infty}\mathcal{P}^{int,o}_{s,j}=1$.

\subsection{Proof of Propositions \ref{beta_e_proposition1} and \ref{beta_e_proposition2}}
\label{proof_proposition_1_2}
\textit{i. Proof of Proposition \ref{beta_e_proposition1}:}
    Substituting $\alpha=4$ and $M_k=2$ into \eqref{ps_k_bound_2} yields $\mathcal{P}_{s,k}^{int,o}(\beta_{e,k})
    =1-\frac{\lambda_e(\xi_k\beta_{e,k})^{-1/2}}
    {\psi_k(1+\xi_k\beta_{e,k})}$.
    Let $x\triangleq(\xi_k\beta_{r,k})
    ^{\frac{1}{2}}$, then we obtain a cubic equation $x^3+x-\frac{\lambda_e}{\psi_k(1-\epsilon)}=0$ from $\mathcal{P}_{s,k}^{int,o}(\beta_{e,k})=\epsilon$.
    Solving it with Cardano's formula \cite{Dunham1990Cardano} completes the proof.

\textit{ii. Proof of Proposition \ref{beta_e_proposition2}:}
    Since $\lim_{M\rightarrow\infty}\left(1+\frac{x}{M}\right)^{-M}= e^{-x}$, we rewrite \eqref{ps_k_bound_2} as $\mathcal{P}_{s,k}^{int,o}(\beta_{e,k})
    =1-\frac{\lambda_e}{\psi_k\beta_{e,k}^{{\delta}}}
    e^{-\frac{1-\phi_k}{\phi_k}\beta_{e,k}}$.
    Let $\frac{1-\phi_k}{\phi_k}\beta_{e,k}\rightarrow x$ and $e^{\frac{\alpha}{2}x}\rightarrow y$, and we obtain $y^y=e^{\theta}$ from $\mathcal{P}_{s,k}^{int,o}(\beta_{e,k})=\epsilon$ with $\theta\triangleq \frac{\alpha}{2}
  \frac{1-\phi_k}{\phi_k}\left(\frac{ (1-\epsilon)\psi_k}{\lambda_e}\right)^{-\frac{\alpha}{2}}$.
  The solution $y$ is given by $y=\frac{\theta}{\mathcal{W}(\theta)}$, which yields the unique root $\beta^{\star}_{e,k}$.


\begin{thebibliography}{99}

\bibitem{Lagrange1997Multitier}
X. Lagrange, ``Multitier cell design," \emph{IEEE Commun. Mag.}, vol. 35, no. 8, pp. 60-64, Aug. 1997.


\bibitem{Xia2010Open}
P. Xia, V. Chandrasekhar and J. G. Andrews, ``Open vs. closed access femtocells in the uplink," \emph{IEEE Trans. Wireless Commun.}, vol. 9, no. 12, pp. 3798-3809, Dec. 2010.

\bibitem{Roche2010Access}
G. de la Roche, A. Valcarce, D. L\'{o}pez-P\'{e}rez, and J. Zhang, ``Access
control mechanisms for femtocells," \emph{IEEE Commun. Mag.}, vol. 48, no. 1, pp. 33-39, Jan. 2010.

\bibitem{Wyner1975Wire-tap}
A. D. Wyner, ``The wire-tap channel," \emph{Bell Syst. Tech. J.}, vol. 54, no. 8, pp. 1355-1387, 1975.

\bibitem{Shafiee2009Towards}
S. Shafiee, N. Liu, and S. Ulukus, ``Towards the secrecy capacity of the
Gaussian MIMO wire-tap channel: The 2-2-1 channel," \emph{IEEE Trans.
Inf. Theory}, vol. 55, no. 9, pp. 4033-4039, Sep. 2009.

\bibitem{Liu2010Multiple}
R. Liu, T. Liu, H. V. Poor, and S. Shamai, ``Multiple-input multiple-output
Gaussian broadcast channels with confidential messages," \emph{IEEE Trans.
Inf. Theory}, vol. 56, no. 9, pp. 4215-4227, Sep. 2010.

\bibitem{Khisti2010Secure2}
A. Khisti and G. W. Wornell, ``Secure transmission with multiple
antennas-Part II: The MIMOME wiretap channel," \emph{IEEE Trans. Inf.
Theory}, vol. 56, no. 11, pp. 5515-5532, Nov. 2010.


\bibitem{Dong2010Improving}
L. Dong, Z. Han, A. P. Petropulu, and H. V. Poor, ``Improving wireless physical layer security via cooperating relays,'' \emph{IEEE Trans. Signal Process.}, vol. 58, no. 3, pp. 1875-1888, Mar. 2010.


\bibitem{HuimingMaga}
H.-M. Wang and X.-G. Xia, ``Enhancing wireless secrecy via cooperation: signal design and optimization,'' \emph{IEEE Commun. Mag.}, vol.53, no. 12, pp. 47-53, Dec. 2015.


\bibitem{Wang2013Joint}
H.-M. Wang, M. Luo, X.-G. Xia, and Q. Yin, ``Joint cooperative beamforming and jamming to secure AF relay systems with individual  power constraint and no eavesdropper's ICSI,'' \emph{IEEE Signal Process. Lett.},  vol.20, no.1, pp.39-42, Jan. 2013.

\bibitem{Zheng2015Outage}
T.-X. Zheng, H.-M. Wang, F. Liu, and M. H. Lee, ``Outage constrained secrecy throughput maximization for DF relay networks," \emph{IEEE Trans. Commun.,} vol. 63, no. 5, pp. 1741-1755, May 2015.

\bibitem{Tekin2008General}
E. Tekin and A. Yener, ``The general Gaussian multiple access and
two-way wire-tap channels: Achievable rates and cooperative jamming,'' \emph{IEEE Trans. Inf. Theory}, vol. 54, no. 6, pp. 2735-2751, Jun. 2008.

\bibitem{Wang2012Distributed}
H.-M. Wang, Q. Yin, and X.-G. Xia, ``Distributed beamforming for
physical-layer security of two-way relay networks," \emph{IEEE Trans.
Signal Process.}, vol. 60, no. 7, pp. 3532-3545, Jul. 2012.

\bibitem{Wang2012Hybrid}
H.-M. Wang, M. Luo, and Q. Yin, ``Hybrid cooperative beamforming and jamming
for physical-layer security of two-way relay networks," \emph{IEEE Trans. Inf. Forensics $\&$ Security}, vol. 8, no. 12, pp. 2007-2020, Dec. 2013.

\bibitem{Mukherjee2014Principles}
A. Mukherjee, S. Fakoorian, J. Huang, and A. Swindlehurst, ``Principles of physical layer security in multiuser wireless networks: A survey," \emph{IEEE Commun. Surveys $\&$ Tutorials}, vol. 16, no. 3, pp. 1550-1573, Mar. 2014.

\bibitem{Goel2008Guaranteeing}
S. Goel and R. Negi, ``Guaranteeing secrecy using artificial noise," \emph{IEEE
Trans. Wireless Commun.}, vol. 7, no. 6, pp. 2180-2189, Jun. 2008.

\bibitem{Zhou2010Secure}
X. Zhou and M. R. McKay, ``Secure transmission with artificial noise over
fading channels: Achievable rate and optimal power allocation," \emph{IEEE
Trans. Veh. Technol.}, vol. 59, no. 8, pp. 3831-3842, Oct. 2010.

\bibitem{Zhang2013Design}
X. Zhang, X. Zhou and M. R. McKay, ``On the design of artificial-noise-aided secure multi-antenna transmission in slow fading channels," \emph{IEEE Trans. Veh. Technol.}, vol. 62, no. 5, pp. 2170-2181, Jun. 2013.


\bibitem{WangANAFF}
H.-M. Wang, T. Zheng, and X.-G. Xia, ``Secure MISO wiretap channels with multi-antenna passive eavesdropper: artificial noise vs. artificial fast fading,'' \emph{IEEE Trans. Wireless Commun.}, vol. 14, no. 1, pp. 94-106, Jan. 2015.

\bibitem{Wang2015Hybrid}
C. Wang, H.-M. Wang, and X.-G. Xia, ``Hybrid opportunistic relaying and jamming with power allocation for secure cooperative networks," \emph{IEEE Trans. Wireless Commun.}, vol. 14, no. 2, pp. 589-605, Feb. 2015.

\bibitem{Deng2015Secrecy}
H. Deng, H.-M. Wang, W. Guo, and W. Wang, ``Secrecy transmission with a helper: to relay or to jam," \emph{IEEE Trans. Inf. Forensics $\&$ Security}, vol. 10, no. 2, pp. 293-307, Feb. 2015.


\bibitem{Haenggi2009Stochastic}
M. Haenggi, J. Andrews, F. Baccelli, O. Dousse, and M. Franceschetti, ``Stochastic geometry and random graphs for the analysis and design of wireless networks," \emph{IEEE J. Select. Areas Commun.}, vol. 27, no. 7, pp. 1029-1046, Sep. 2009.

\bibitem{Zhou2011Secure}
X. Zhou, R. K. Ganti, and J. G. Andrews, ``Secure wireless network connectivity with multi-antenna transmission," \emph{IEEE Trans. Wireless Commun.}, vol. 10, no. 2, pp. 425-430, Feb. 2011.

\bibitem{Zheng2014Transmission}
T.-X. Zheng, H.-M. Wang, and Q. Yin, ``On transmission secrecy outage of a multi-antenna system with randomly located eavesdroppers," \emph{IEEE Commun. Lett.}, vol. 18, no. 8, pp. 1299-1302, Aug. 2014.

\bibitem{Ghogho2011Physical}
M. Ghogho and A. Swami, ``Physical-layer secrecy of MIMO communications in the presence of a Poisson random field of eavesdroppers," in \emph{Proc. IEEE ICC Workshops}, Jun. 2011, pp. 1-5.

\bibitem{Zheng2015Multi}
T.-X. Zheng, H.-M. Wang, J. Yuan, D. Towsley, and M. H. Lee, ``Multi-antenna transmission with artificial noise against randomly distributed eavesdroppers", \emph{IEEE Trans. on Commun.}, vol. 63, no. 11, pp. 4347-4362, Nov. 2015.

\bibitem{Zheng2015Multi2}
T.-X. Zheng and H.-M. Wang, ``Optimal power allocation for artificial noise under imperfect CSI against spatially random eavesdroppers,'' \emph{IEEE Trans. on Vehi. Techn.}, accepted to appear, 2015, DOI: 10.1109/TVT.2015.2513003.

\bibitem{Zhou2011Throughput}
X. Zhou, R. Ganti, J. Andrews, and A. Hj{\o}rungnes, ``On the throughput
cost of physical layer security in decentralized wireless networks," \emph{IEEE
Trans. Wireless Commun.}, vol. 10, no. 8, pp. 2764-2775, Aug. 2011.

\bibitem{Zhang2013Enhancing}
X. Zhang, X. Zhou, and M. R. McKay, ``Enhancing secrecy with multi-antenna transmission in wireless ad hoc networks," \emph{IEEE Trans. Inf. Forensics and Security}, vol. 8, no. 11, pp. 1802-1814, Nov. 2013.

\bibitem{Wang2013Physical}
H. Wang, X. Zhou, and M. C. Reed, ``Physical layer security in cellular networks:
A stochastic geometry approach," \emph{IEEE Trans. Wireless Commun.}, vol. 12, no. 6, pp. 2776-2787, June 2013.

\bibitem{Geraci2014Physical}
G. Geraci, H. S. Dhillon, J. G. Andrews, J. Yuan, and I. B. Collings, ``Physical layer security in downlink multi-antenna cellular networks," \emph{IEEE Trans. Commun.}, vol. 62, no. 6, pp. 2006-2021, June 2014.

\bibitem{Yang2015Safeguarding}
N. Yang, L. Wang, G. Geraci, M. Elkashlan, J. Yuan, and M. D. Renzo, ``Safeguarding 5G wireless communication networks using physical layer security," \emph{IEEE Commun. Mag.}, vol. 53, no. 4, pp. 20-27, Apr. 2015.


\bibitem{Lv2015Secrecy}
T. Lv, H. Gao, and S. Yang, ``Secrecy transmit beamforming for heterogeneous networks," \emph{IEEE J. Sel. Areas Commun.}, vol. 33, no. 6, pp. 1154-1170, June 2015.

\bibitem{Dhillon2012Modeling}
H. S. Dhillon, R. K. Ganti, F. Baccelli, and J. G. Andrews, ``Modeling and analysis of K-tier downlink heterogeneous cellular networks," \emph{IEEE J. Sel. Areas Commun.}, vol. 30, no. 3, pp. 550-560, Apr. 2012.


\bibitem{Haenggi08Distance}
M. Haenggi, ``On distances in uniformly random networks," \emph{IEEE Trans. Inf. Theory}, vol. 51. no. 10, pp. 3584-3586, Oct. 2005.

\bibitem{Andrews2011Tractable}
J. G. Andrews, F. Baccelli, and R. K. Ganti, ``A tractable approach to
coverage and rate in cellular networks," \emph{IEEE Trans. Commun.}, vol. 59, no. 11, pp. 3122-3134, Nov. 2011.

\bibitem{Jo2012Heterogeneous}
H.-S. Jo, Y. J. Sang, P. Xia, and J. G. Andrews, ``Heterogeneous cellular networks with
flexible cell association:
A comprehensive downlink SINR analysis," \emph{IEEE Trans. Wireless Commun.}, vol. 11, no. 10, pp. 3484-3495, Oct. 2012.

\bibitem{Hunter08Transmission}
A. M. Hunter, J. G. Andrews, S. Weber, ``Transmission capacity of
ad hoc networks with spatial diversity," \emph{IEEE Trans. Wireless Commun.}, vol. 7, no. 12, pp. 5058-5071, Dec. 2008.

\bibitem{Louie2011Open}
R. H. Y. Louie, M. R. McKay, and I. B. Collings, ``Open-loop spatial multiplexing and diversity communications in ad hoc networks," \emph{IEEE
Trans. Inf. Theory}, vol. 57, no. 1, pp. 317-344, Jan. 2011.

\bibitem{Li2014Throughput}
C. Li, J. Zhang, and K. B. Letaief, ``Throughput and energy efficiency analysis of small cell networks with multi-antenna base stations," \emph{IEEE Trans. Wireless Commun.}, vol. 13, no. 5, pp. 2505-2517, May 2014.

\bibitem{Alzer1997Mathematics}
H. Alzer, ``On some inequalities for the incomplete gamma function," \emph{Mathematics of Computation}, vol. 66, no. 218, pp. 771-778, Apr. 1997.


\bibitem{Dunham1990Cardano}
W. Dunham, ``Cardano and the solution of the cubic," Ch. 6 in \emph{Journey through Genius: The Great Theorems of Mathematics}, pp. 133-154. John Wiley and Sons, Inc., 1990.

\bibitem{Olver2010NIST}
F. W. J. Olver, D. W. Lozier, R. F. Boisvert, and C. W. Clark, \emph{NIST Handbook of Mathematical Functions}, Cambrige, USA: Cambrige Univ. Press, 2010.


\bibitem{Stoyan1996Stochastic}
D. Stoyan, W. Kendall, and J. Mecke, \emph{Stochastic Geometry and its
Applications, 2nd ed}. John Wiley and Sons, 1996.

\bibitem{Gradshteyn2007Table}
I. S. Gradshteyn, I. M. Ryzhik, A. Jeffrey, D. Zwillinger, and S. Tech-
nica, \emph{Table of Integrals, Series, and Products, 7th ed}. ~New York:
Academic Press, 2007.

\end{thebibliography}
\end{document}